\documentclass[12pt]{article}
\usepackage[latin1]{inputenc}
\usepackage[english]{babel}
\usepackage{amsfonts,amsmath,mathrsfs,relsize}
\usepackage[dvips]{graphicx}
\usepackage{appendix}
\usepackage[round,]{natbib}
\usepackage{verbatim}
\usepackage{booktabs, array}
\usepackage{mwe, psfrag}    % loads »blindtext« and »graphicx«
\usepackage{subfigure}
\usepackage{color, caption}
\usepackage{mathtools}
\usepackage{dashrule}
\usepackage{soul}

\usepackage[lined,boxed,commentsnumbered]{algorithm2e}
\usepackage{multirow}
\usepackage[top=2.5cm,left=2.5cm,right=2.5cm,bottom=2.8cm]{geometry}
\usepackage{hyperref}

\DeclarePairedDelimiter{\ceil}{\lceil}{\rceil}
\def\T{{\footnotesize {^{_{\sf T}}}}} 
\newcommand{\Real}{{\rm I}\negthinspace {\rm R}}

\newcommand{\dd}{\text{d}}
\newtheorem{theorem}{Theorem}[section]

\newcommand{\qedwhite}{\hfill \ensuremath{\Box}}

\newenvironment{proof}[1][Proof]{\begin{trivlist}
\item[\hskip \labelsep {\bfseries #1}]}{\end{trivlist}}

\newenvironment{scheme}[1][htb]
  {% Update algorithm name
   \begin{algorithm}[#1]%
  }{\end{algorithm}}
  
\begin{document}
%%%%%%%%%%%%%%%%%%%%%%%%%%%%%%%%%%%%%%%%%%%%%%%%%%%%%%%%%%%%%%%%

\title{Improved Laplace Approximation\\ for Marginal Likelihoods}

\author{Erlis Ruli$^*$, Nicola Sartori and  Laura Ventura\\
{\it {\small Department of Statistics, University of Padova, Italy}} \\
{\tt {\small $^*$ruli@stat.unipd.it, sartori@stat.unipd.it, ventura@stat.unipd.it}} }
\maketitle

%%%%%%%%%%%%%%%%%%%%%%%%%%%%%%%%%%%%%%%%%%%%%%%%%%%%%%%%%%%%%%%%
\begin{abstract}
Statistical applications often involve the calculation  of intractable multidimensional integrals. The Laplace formula is widely used to approximate such integrals. However, in high-dimensional or small sample size problems, the shape of the integrand function may be far from that of the Gaussian density, and thus the standard Laplace approximation can be inaccurate. We propose an improved Laplace approximation that reduces the asymptotic error of the standard Laplace formula by one order of magnitude, thus leading to third-order accuracy. We also show, by means of practical examples of various complexity, that the proposed method is extremely accurate, even in high dimensions, improving over the standard Laplace formula. Such examples also demonstrate that the accuracy of the proposed method is comparable with that of other existing methods, which are computationally more demanding. An \texttt{R} implementation of the improved Laplace approximation is \st{also} provided through the \texttt{R} package \texttt{iLaplace} available on \texttt{CRAN}. 
\end{abstract}

%%%%%%%%%%%%%%%%%%%%%%%%%%%%%%%%%%%%%%%%%%%%%%%%%%%%%%%%%%%%%%%%


%%%%%%%%%%%%%%%%%%%%%%%%%%%%%%%%%%%%%%%%%%%%%%%%%%%%%%%%%%%%%%%%
\section{Background}
%%%%%%%%%%%%%%%%%%%%%%%%%%%%%%%%%%%%%%%%%%%%%%%%%%%%%%%%%%%%%%%%
\label{sec:1}
Statistical applications often involve the evaluation of finite integrals of the form
\begin{equation}
I_n = \int_{\Real^d}^{}e^{-h_n(x)}\,\text{d}x\,\label{eq:target},
\end{equation}
where $h_n(x)$ is a smooth and concave real function, with $x$ a $d$-dimensional real vector, indexed by $n>0$. For instance, in Bayesian analyses, $-h_n(\cdot)$ may be the log-likelihood or the log-posterior kernel and \eqref{eq:target} is the Bayesian marginal likelihood or the posterior normalising constant. Furthermore, in Generalized Linear Mixed Models (GLMM) $-h_n(\cdot)$ may represent the log-likelihood plus the log-density of the random effects.  In this case, \eqref{eq:target} gives the marginal likelihood for the parameters $(\theta, \theta_u)$, which can be generally written as
\begin{eqnarray}\notag
L(\theta, \theta_u;y) &=& \int L(\theta;u,y)f(u;\theta_u)\,\dd u\\\notag
 &=& \int  \exp\{\log L(\theta;u,y) + \log f(u;\theta_u)\}\,\dd u\\
 & = & \int e^{-h_n(u;\theta,\theta_u, y)}\,\dd u\,,
 \label{eq:mlGLMM}
\end{eqnarray}
where $f(u;\theta_u)$ is the density of the random effects indexed by the parameter $\theta_u$, and $L(\theta;u,y)$ is the likelihood for $\theta$ based on the conditional density of $y$ given $u$. The quantity $n$ is related to the information in the sample, and is often the sample size.

Integral \eqref{eq:target} is frequently intractable but it can be approximated by several methods \citep[see, e.g.,][]{evans2000}. %Numerical integration via quadrature rules is typically accurate but, due to the curse of dimensionality, it is feasible only for low-dimensional integrals, i.e. $d<10$ \citep{evans1995}. Monte Carlo methods are widely applicable to general functions $h_n(x)$, especially for large $d$, provided $I_n$ is finite. However, Monte Carlo integration is typically simulation consistent, in the sense that it reaches the true value as the number of simulations goes to infinity.  Alternatively, when $-h_n(x)$ has a single mode, which becomes increasingly dominant for large values of $n$, methods based on asymptotic expansions, such as Laplace's formula, are computationally more convenient. Asymptotic expansions around the mode give analytical approximations of $I_n$ which do not require tuning, and which are consistent for $n\to \infty$. Hence, unlike Monte Carlo approximations, the accuracy of asymptotic methods is bounded by some function of the given $n$. Nevertheless, when $-h_n(x)$ is single peaked, even for moderate values of $n$ asymptotic approximations may outperform Monte Carlo methods for a given computational cost \citep[see, e.g.,][]{rue2009approximate,ruli2014marginal}. This makes asymptotic approximations particularly appealing in applications.
Here, we focus on the Laplace approximation (see, e.g., \citealp[][]{bleistein1986}, Chap. 8 and \citealp[][]{small}, Chap. 6). Let $\hat x=(\hat{x}_1,\ldots,\hat{x}_d)$ be the unique minimum of $h(\cdot)$, where to ease notation hereafter we drop $n$ from $I_n$, $h_n(\cdot)$ and related quantities. In addition, we assume that the Hessian matrix of $h(\cdot)$ at $\hat{x}$,  i.e.
\[
\hfill\hat V  =  V(\hat{x}) = \frac{\partial^2 h(x)}{\partial x\partial x^\T}\bigg\vert_{x = \hat{x}}\,,\hfill
\]
% that $\hat x$ is the only point for which
%\[
%\hfill\nabla h(x) = \left(\frac{\partial h(x)}{\partial x_1},\ldots,\frac{\partial h(x)}{\partial x_d}\right) = 0\hfill,
%\]
%and that 
is positive definite. The Laplace approximation of \eqref{eq:target} is second-order accurate, i.e., $I = \hat{I}^{\,\texttt{L}}\{1+O(n^{-1})\}$, with 
\begin{equation}
\hat{I}^{\,\texttt{L}} = (2\pi)^{d/2}|\hat{V}|^{-1/2}H(\hat{x}),
\label{eq:lap}
\end{equation}
where $H(\cdot) = \exp\{-h(\cdot)\}$; see, e.g., \citet[][p. 335]{bleistein1986}.

%Consider a sample $y=(y_1,\ldots,y_n)$ of size $n$ from a model $p(y;\theta)$, $\theta\in\Theta\subseteq \Real^d$. In Bayesian applications, $-h(\cdot)$ can be the logarithm of the posterior distribution $\pi(\theta|y)\propto L(\theta;y)\pi(\theta)$, where $L(\theta;y)$ is the likelihood function for $\theta$ and $\pi(\theta)$ is the prior distribution. In this case, \eqref{eq:target} gives the posterior normalizing constant that can be used, for instance, to compute Bayes Factors (BFs). In GLMM, $-h(\cdot)$ is the logarithm of the likelihood, conditional on the random effects, times the density of the latter. In this case, the integration is performed with respect to the random effects and \eqref{eq:target} gives the marginal likelihood of the fixed parameters based on $y$.

The Laplace approximation is widely used both in the Bayesian framework for approximating posterior densities and posterior moments \citep[see, e.g.,][]{tk1986,rue2009approximate} or Bayes Factors \citep[see, e.g.,][]{kass1995bayes}, and in the frequentist framework for integrating out random effects in GLMM \citep[see, e.g.,][]{breslow1993} or to compute marginal likelihoods in group models (\citealp[see, e.g., ][Sect. 2.8]{cox1994inference}; \citealp{Pace20063539}). In addition, it has also been used to approximate hypergeometric functions of matrix arguments \citep{butler2002laplace}. %In addition, \cite{guihenneuc2005laplace} use Laplace's formula within a Markov chain Monte Carlo (MCMC) framework, in which the nuisance parameters are integrated out by the Laplace method and the parameters of interest are integrated by MCMC. 
Moreover, \cite{ruli2014marginal} propose a simulation algorithm which draws posterior samples by inverting the approximate cumulative distribution function based on the Laplace approximation for marginal posterior densities. Lastly, \cite{martino2011} and \cite{rizopoulos2009fully} apply the Laplace method in the context of survival analysis and joint modelling of survival and longitudinal data, respectively. 

In the standard asymptotic setting with $d$ fixed and $n\to\infty$, the Laplace approximation \eqref{eq:lap} is second-order accurate. On the other hand, if also $d$ is large, the asymptotic expansion requires more terms in order to achieve the same accuracy as in lower-dimensions. For instance, when \eqref{eq:target} factorises as a product of $d$ scalar identical integrals, the relative error is of order $O(d/n)$ \citep[Sect.\ 6.9]{small}. However, in practice both $d$ and $n$ are fixed, and when $d$ is large relatively to $n$ it may be necessary to improve the accuracy of the standard Laplace method. Moreover, an unappealing feature of the Laplace approximation is that it does not account for skewness or kurtosis in the integrand function. Therefore, when the shape of the integrand is far from that of the Gaussian density, which can happen especially in high-dimensional or in small sample size problems, the standard Laplace approximation can be severely inaccurate. Example~\ref{sec:multT} in Sect. 3 shows an example in which the Laplace approximation fails dramatically, and with inaccuracy that deteriorates with increasing dimensionality. 

A possible way to improve the Laplace approximation is through the inclusion of higher-order derivatives of $h(\cdot)$ in the Taylor expansions. \cite{lindley1980} uses this idea in a Bayesian context. \cite{raudenbush2000maximum} propose a higher-order Laplace approximation for GLMM, by considering derivatives of $h(\cdot)$ up to the the sixth order. \cite{Pace20063539} use a similar approach for approximating marginal likelihoods in group models. However, when $d>1$, the computation of higher-order derivatives can be tedious.  Similar strategies are pursued by the Bayesian Bartlett correction proposed by \cite{diciccio1997computing}, and by the corrected Laplace approximation of  \cite{shun1995laplace}. However, the former involves posterior expectations, which in practice must be approximated through Monte Carlo methods and the latter solution is designed for situations, such as models with crossed random effects, in which the standard Laplace approximation may not be asymptotically valid. Another improvement of the standard Laplace approximation is proposed by \cite{nott2009generalization}, in which \eqref{eq:target} is approximated by a product of scalar blocks, after a preliminary variable transformation to achieve approximate orthogonality. %Although the proposal of \cite{nott2009generalization} tends to work slightly better than the standard Laplace approximation, we show (see Example~\ref{sec:multT}) that our method provides exceedingly more accurate approximations.

In this paper we propose an improved Laplace approximation for integrals of the form \eqref{eq:target} that, unlike the standard Laplace formula, can account for skewness and non-Gaussian tails in the integrand function. Moreover, we show that the proposed method has relative approximation error of order $O(n^{-3/2})$, in a standard asymptotic setting in which the sample size $n$ diverges and $d$ is fixed. The core idea of the proposed method is to build an approximation of the normalised integrand through sequential and re-normalised ratios of Laplace approximations. Finally, an approximation of the target integral \eqref{eq:target} is obtained indirectly by the ratio of the un-normalised integrand over the approximation of the normalised integrand, both evaluated at a specific point $x$. Essentially, the proposed approximation of \eqref{eq:target} can be written as $I^{\texttt{iL}} = \hat{c}\,I^{\texttt{L}}$, where $\hat{c}>0$ is the improvement over the standard Laplace approximation, and $I = I^{\texttt{iL}}\{1+O(n^{-3/2})\}$.

%This result is shown by empirical and theoretical results reported in the Appendix. In addition, in a non-asymptotic example in which the Laplace approximation fails miserably, our method gives practically exact answers.

%The main idea of the proposed method is to first split the multidimensional integral \eqref{eq:target} in $d$ conditional blocks of scalar integrals, which are approximated by Laplace's formula for marginal densities \cite{tk1986}. Then each approximated block is renormalised numerically 
%and is obtained as the ratio between the unormalised kernel and an improved Laplace approximation of the normalized kernel, both evaluated at a fixed value of $x$. 

Compared to the standard Laplace approximation, the proposed method requires repeated conditional minimisations and repeated evaluations of the log-integrand function and its Hessian matrix. Conditional minimisations can be computationally demanding in high dimensions. Therefore, an alternative version is introduced, which uses approximate conditional minima obtained through a first order Taylor series expansion around the global minimum. This alternative version reduces the computational time while keeping comparable accuracy with respect to the original version. Nevertheless, the most demanding task is the computation of the global minimum, which is a requirement also for the standard Laplace method. %More details on implementation issues are discussed in Sections 3, 4 and 5.
 
The rest of the article is structured as follows. Section 2 introduces the improved Laplace approximation. Section 3 illustrates the method in examples in which comparison with alternative approximations are also given. Section 4 concludes with some final remarks.

%%%%%%%%%%%%%%%%%%%%%%%%%%%%%%%%%%%%%%%%%%%%%%%%%%%%%%%%%%%%%%%%
%\section{Background}
%%%%%%%%%%%%%%%%%%%%%%%%%%%%%%%%%%%%%%%%%%%%%%%%%%%%%%%%%%%%%%%%%
%\label{sec:1}
%In order to ease notation, in the following we drop $n$ form $I_n$, $h_n(\cdot)$ and related quantities. Hence, let $h(x)$ be a smooth function of $x=(x_1,\ldots,x_d)$, at least twice differentiable and with unique minimum at $\hat x=(\hat{x}_1,\ldots,\hat{x}_d)$. In addition, we assume that $\hat x$ is the only point for which
%\[
%\hfill\nabla h(x) = \left(\frac{\partial h(x)}{\partial x_1},\ldots,\frac{\partial h(x)}{\partial x_d}\right) = 0\hfill,
%\]
%and that the Hessian matrix of $h(x)$ at $\hat{x}$,  i.e.
%\[
%\hfill\hat V  =  V(\hat{x}) = \frac{\partial^2 h(x)}{\partial x\partial x^\T}\bigg\vert_{x = \hat{x}}\,,\hfill
%\]
%is positive definite. More precise regularity conditions are provided in the Appendix. Finally, let $H(\cdot) = \exp\{-h(\cdot)\}$. The standard Laplace approximation of \eqref{eq:target} can be written as $I = \hat{I}^{\,\texttt{L}}\{1+O(n^{-1})\}$, with 
%\begin{equation}
%\hat{I}^{\,\texttt{L}} = (2\pi)^{d/2}|\hat{V}|^{-1/2}H(\hat{x});
%\label{eq:lap}
%\end{equation}
%see, e.g., \citet[][p. 335]{bleistein1986}.

\section{The improved Laplace approximation}
\label{sec:2}
Let $p(x) = H(x)/I$ be the density function which corresponds to the kernel $H(x) = \exp\{-h(x)\}$ with normalising constant $I$. By the identity
\begin{equation}
\hfill I = \frac{H(x)}{p(x)}\,,\hfill
\label{eq:malikid}
\end{equation}
if $p(x)$ is known then $I$ is readily available, for an arbitrary $x$. Alternatively, if a suitable estimate $\hat{p}(x)$ of ${p}(x)$ is available, \eqref{eq:malikid} provides an estimate $\hat{I}$ of $I$, given by $H(x)/\hat{p}(x)$. For instance, \eqref{eq:malikid} has been used to estimate Bayesian marginal likelihoods in MCMC settings \citep{chib1995,chib2001,hsiao2004bayesian}, and {to approximate hidden} Gaussian Markov random fields \citep{rue2004approximating,rue2009approximate}. 

While \eqref{eq:malikid} holds for any $x$,  it is advisable to locate such a point at a high density region \citep[][]{chib1995}. One possibility is to choose $x=\hat{x}$, which may be also convenient from a computational point of view. {In MCMC settings, \cite{hsiao2004bayesian} show that coordinate points other than $\hat{x}$ may improve the approximation error. However, locating such points can be computationally intensive.}

Let $x_{1:q} = (x_1,\ldots,x_q)$ be the first $q$ and $x_{q+1:d}$ the last $d-q$ components of $x$ ($q<d$). Moreover, let $\hat{x}_{x_1}$ be the conditional minimum of $h(\cdot)$ with $x_1$ fixed and let $\hat{x}_{\hat{x}_{1:q}, x_{q+1}}$ be the conditional minimum with $x_{1:q}$ fixed at $\hat{x}_{1:q}$ and $x_{q+1}$ fixed. We require that $h(\cdot)$ satisfies the usual regularity conditions for the validity of the Laplace approximation \citep[see, e.g.,][]{kass1990validity}.

Write $p(x)$ as
\begin{eqnarray}\nonumber
p(x) &=& p_{X_1}(x_1)\times p_{X_2|X_1}(x_2|x_1)\times\cdots \times p_{X_d|X_{1:d-1}}(x_d|x_{1:d-1})\\
& = & \frac{\int_{\Real^{d-1}} H(x)\,\text{d}x_{2:d}}{\int_{\Real^d} H(x)\,\text{d}x}\times\frac{\int_{\Real^{d-2}} H(x)\,\text{d}x_{3:d}}{\int_{\Real^{d-1}} H(x)\,\text{d}x_{2:d}}\times\cdots\times
\frac{H(x)}{\int_{\Real} H(x)\,\text{d}x_{d}}\,.
\label{eq:expanded}
\end{eqnarray}
An improved approximation of $p(x)$ can be obtained by approximating the integrals of each ratio on the right hand side of \eqref{eq:expanded} through the Laplace formula. Specifically, the Laplace approximation of the marginal density $p_{X_1}(x_1)$ is
\begin{equation}
\hat{p}_{X_1}(x_1) = \frac{H(x_1,\hat{x}_{x_{1}})}{H(\hat{x})}\left\{\frac{|\hat V|}{2\pi|V_{2:d}(x_1,\hat{x}_{x_1})|}\right\}^{1/2}\,,\label{eq:iLAmarg}
\end{equation}
where $V_{2:d}(\cdot)$ is the block (2:$d$, 2:$d$) of $V(\cdot)$. This result is due to \cite{tk1986}. %Upon numerical normalisation, this marginal density is third-order accurate in regions of $\hat{x}_1$ that shrink in diameter at the rate $n^{-1/2}$. Indeed, let $$\hat{p}^*(x_1) = \frac{\hat{p}(x_1)}{\int\hat{p}(x_1)\,\text{d}x_1}.$$ Then, we have $p(x_1) =\hat{p}^*(x_1)\{1+O(n^{-3/2})\}$ \citep[see also][]{reid2003}. From this result we readily have a third order-accurate approximation of $p(\hat{x}_1)$, given by $\hat{p}^*(\hat{x}_1)$. The most relevant aspect for our purpose is that this approximation of $p(x_1)$ does not force normality of the integrand in the direction $x_1$, although it does for the other components of $x$. However, as showed by \cite{tk1986}, the methods tends te be very accurate in practice.
%Consider now the second factor on the right-hand side of \eqref{eq:expanded}, i.e. the conditional density $p(x_2|x_1)$. As for $p(x_1)$, we can apply the standard Laplace method to the numerator and denominator of $p(x_2|x_1)$ and obtain the quantity
%\begin{equation}
%\hat{p}(x_2|x_1) = \frac{H(x_{1:2},\hat{x}_{x_{1:2}})\{(2\pi)^{d-2} |V_{3:d}(x_{1:2},\hat{x}_{x_{1:2}})|\}^{-1/2}\{1+O(n^{-1})\}}{H(x_1,\hat{x}_{x_1})\{(2\pi)^{d-1} |V_{2:d}(x_{1},\hat{x}_{x_{1}})|\}^{-1/2}\{1+O(n^{-1})\}}\,,\label{eq:condtmp}
%\end{equation}
%Using results and assumptions in \cite{kass1990validity}, it is possible to show that this conditional density has relative error of order $O(n^{-1})$. But in general, to achieve third-order accuracy, both the numerator and the denominator must be renormalized numerically. 
For the $q$th conditional density in \eqref{eq:expanded} ($2\leq q < d$), we apply the Laplace approximation to the numerator and the denominator and obtain
\begin{eqnarray}
\hat{p}_{X_q|X_{1:q-1}}(x_q|x_{1:q-1}) = \frac{\hat{p}_{X_{1:q}}(x_{1:q})}{\hat{p}_{X_{1:q-1}}(x_{1:q-1})}\,.\label{eq:condfree}
%\hat{p}(x_q|x_{1:q-1}) = \frac{H(x_{1:q},\hat{x}_{x_{1:q}})\{(2\pi)^{d-q} |V_{q:d}(x_{1:q},\hat{x}_{x_{1:q}})|\}^{-1/2}\{1+O(n^{-1})\}}{H(x_{1:q-1},\hat{x}_{x_{1:q-1}})\{(2\pi)^{d-q+1} |V_{q-1:d}(x_{1:q-1},\hat{x}_{x_{1:q-1}})|\}^{-1/2}\{1+O(n^{-1})\}}\,.\label{eq:condfree}
\end{eqnarray}
Finally, the conditional density $p_{X_d\vert X_{1:d-1}}(\cdot)$, approximated by applying the univariate Laplace method to the integral in the denominator, is
\begin{eqnarray}
\hat{p}_{X_d|X_{1:d-1}}(x_d|x_{1:d-1}) = \frac{H(x)}{H(x_{1:d-1},\hat{x}_{x_{1:d-1}})}\left\{\frac{V_{d:d}(x_{1:d-1},\hat{x}_{x_{1:d-1}})}{2\pi} \right\}^{1/2}\,,\label{eq:condfree2}
\end{eqnarray}
where $V_{d:d}(\cdot)$ is the $d$th element of the diagonal of $V(\cdot)$.
Using results and under the assumptions of \cite{kass1990validity} and of \cite{tk1986}, it is possible to show that \eqref{eq:iLAmarg}, \eqref{eq:condfree} and \eqref{eq:condfree2} have overall relative error of order $O(n^{-1})$. Recalling that, to use identity \eqref{eq:malikid} we only need an approximation $\hat{p}(\hat{x})$ of $p(\hat{x})$, we might be tempted to take as $\hat{p}(\hat{x})$ the product of \eqref{eq:iLAmarg} times \eqref{eq:condfree} (for $2\leq q <d$) times \eqref{eq:condfree2}, all evaluated at $\hat{x}$. However, such a product, when replaced in \eqref{eq:malikid}, reproduces exactly \eqref{eq:lap}, the Laplace approximation of $I$. 

To achieve third-order accuracy we propose to re-normalise numerically \eqref{eq:iLAmarg}, \eqref{eq:condfree} and \eqref{eq:condfree2}. Re-normalisation of \eqref{eq:condfree} and \eqref{eq:condfree2} entails the evaluation of multi-dimensional numerical integrations. While this is true in general, in our case we only need an approximation for $p(\hat{x})$ and therefore it is still possible to re-normalise \eqref{eq:condfree} and \eqref{eq:condfree2} by using only scalar numerical integration. The key point is to fix all the conditioning variables at the corresponding modal values prior to the re-normalisations, as explained in Scheme~\ref{alg:iLaplace}.

The product of the re-normalised versions of \eqref{eq:iLAmarg}, \eqref{eq:condfree} and \eqref{eq:condfree2}, evaluated at $\hat{x}$, gives a third-order approximation of $p(\hat{x})$, as shown by the following theorem.

\begin{scheme}
 \caption[]{\label{alg:iLaplace} Pseudo-code description of the improved Laplace method.}
  \SetAlgoLined
  %\KwResult{A sample $(\theta^{(1)},\ldots,\theta^{(m)})$ from $\pi_\epsilon(\theta|t(y^{\text{obs}}))$}
  %\KwData{A statistic $t(\cdot)$ threshold $\epsilon$}
  \begin{description}
  \item[Step 1] Compute $\hat{c}_1$, the normalising constant of \eqref{eq:iLAmarg}\; 
  \item[Step 2] For each $q$ ($2\leq q<d$), compute $\hat{c}_q$, the normalising constant of \eqref{eq:condfree} with all the conditioning variables fixed at their modal values\;
  \item[Step 3] Compute $\hat{c}_d$, the normalising constant of \eqref{eq:condfree2} with all the conditioning variables fixed at the corresponding modal values\;
  \item[Step 4] Set $\hat{p}^{\texttt{iL}}(\hat{x}) = \frac{\hat{p}(\hat{x}_1)}{\hat{c}_1}\left\{\prod_{q=2}^{d-1} \frac{\hat{p}(\hat{x}_q\vert \hat{x}_{1:q-1})}{\hat{c}_q}\right\}\frac{\hat{p}(\hat{x}_d\vert \hat{x}_{1:d-1})}{\hat{c}_d}$ as the improved Laplace approximation of $p(\hat{x})$\;
  \item[Step 5] Finally, get the improved Laplace approximation $I^{\texttt{iL}} = H(\hat{x})/\hat{p}^{\texttt{iL}}(\hat{x})$ of $I$.
  \end{description}
\end{scheme}

\begin{theorem}
\label{th:theorem1}
Under the assumptions of \cite{kass1990validity} for the regularity of the Laplace approximation, the improved Laplace approximation of $p(\hat{x})$ has third-order accuracy, i.e.
\[
p(\hat{x})= \hat p^{\texttt{iL}}(\hat{x}) \{1+O(n^{-3/2})\}
\,.\]
\end{theorem}
\begin{proof}
The first step is to show that the approximation error of \eqref{eq:iLAmarg} and \eqref{eq:condfree} holds uniformly. The uniformity of \eqref{eq:iLAmarg} has been already shown by \citeauthor[][]{kass1990validity} (1990, Theorem 6). Furthermore, on the basis of Theorem 6 of \cite{kass1990validity}, we can show that also \eqref{eq:condfree} holds uniformly. This is immediate as \eqref{eq:condfree} is the ratio of the Laplace approximation of the marginal density of $X_{1:q}$ over that of $X_{1:q-1}$ ($2\leq q < d$), both with relative error of order $O(n^{-1})$ holding uniformly. Hence the error in \eqref{eq:condfree} is also uniform. 

\cite{tk1986} show that, after numerical re-normalisation, \eqref{eq:iLAmarg} has error of order $O(n^{-3/2})$. This is because the $O(n^{-1})$ term, when uniform, gets absorbed into the normalising constant. To complete the proof we need to show that also \eqref{eq:condfree}, after numerical re-normalisation with respect to $x_q$ and with the conditioning variables fixed at the corresponding modal values, has error of order $O(n^{-3/2})$. To prove this, let $\hat{p}^*_{X_q|X_{1:q-1}}(x_q|\hat{x}_{1:q-1})$ be the re-normalised approximate conditional density of $X_q|X_{1:q-1}$ with the conditioning variables fixed at the modal values, i.e.
\[
\hat{p}^*_{X_q|X_{1:q-1}}(x_q|\hat{x}_{1:q-1}) = \frac{c_q^{-1}\hat{p}_{X_{1:q}}(\hat{x}_{1:q-1},x_{q})}{c_{q-1}^{-1}\hat{p}_{X_{1:q-1}}(\hat{x}_{1:q-1})}\,,
\]
where $c_q = \int_{\Real^q} \hat{p}_{X_{1:q}}(x_{1:q})\,\dd x_{1:q}$ and $c_{q-1} = \int_{\Real^{q-1}} \hat{p}_{X_{1:q-1}}(x_{1:q-1})\,\dd x_{1:q-1}$. This re-normalised approximate conditional density is third-order accurate, i.e.
\begin{align*}
p_{X_{q}|X_{1:q-1}}(x_q|\hat{x}_{1:q-1}) \,=\, & \frac{c_{q}^{-1}\hat{p}_{X_{1:q}}(\hat{x}_{1:q-1},x_q)\{1+O(n^{-3/2})\}}{c_{q-1}^{-1}\hat{p}_{X_{1:q-1}}(\hat{x}_{1:q-1})\{1+O(n^{-3/2})\}}  \\
 \,=\, &\frac{c_{q}^{-1}\hat{p}_{X_{1:q}}(\hat{x}_{1:q-1},x_q)}{c_{q-1}^{-1}\hat{p}_{X_{1:q-1}}(\hat{x}_{1:q-1})} \{1+O(n^{-3/2})\}\\
 \,=\, & \hat{p}^*_{X_q|X_{1:q-1}}(x_q|\hat{x}_{1:q-1}) \{1+O(n^{-3/2})\}\,.
\end{align*}
{However,} note that there is no need to compute $c_{q-1}$ and $c_q$, because
\begin{align*}
\hat{p}^*_{X_q|X_{1:q}}(x_q|\hat{x}_{1:q-1}) \,=\,& \frac{\frac{c_q^{-1}\hat{p}_{X_{1:q}}(\hat{x}_{1:q-1},x_{q})}{c_{q-1}^{-1}\hat{p}_{X_{1:q-1}}(\hat{x}_{1:q-1})}}{\mathop{\mathlarger{\int_{\Real}}}\frac{c_q^{-1}\hat{p}_{X_{1:q}}(\hat{x}_{1:q-1},x_{q})}{c_{q-1}^{-1}\hat{p}_{X_{1:q-1}}(\hat{x}_{1:q-1})}\,\dd x_q} \\
 \,=\, & \frac{\hat{p}_{X_{1:q}}(\hat{x}_{1:q-1},x_{q})/\hat{p}_{X_{1:q-1}}(\hat{x}_{1:q-1})}{\int_{\Real} \hat{p}_{X_{1:q}}(\hat{x}_{1:q-1},x_{q})/\hat{p}_{X_{1:q-1}}(\hat{x}_{1:q-1})\,\dd x_q} \\
 \,= \, & \frac{\hat{p}_{X_{1:q}}(\hat{x}_{1:q-1},x_{q})/\hat{p}_{X_{1:q-1}}(\hat{x}_{1:q-1})}{\hat{c}_q}.
\end{align*}
That is, we need only to re-normalise \eqref{eq:condfree} with the conditioning variables fixed prior to integration at their modal values, i.e. to compute $\hat{c}_q$. {Note that, after numerical re-normalisation with the conditioning variables fixed, approximation \eqref{eq:condfree2} becomes exact.} \qedwhite
\end{proof}

Finally, the replacement of $p(\hat{x})$ with $\hat{p}^{\texttt{iL}}(\hat{x})$ in \eqref{eq:malikid} delivers the improved Laplace approximation of \eqref{eq:target}. Or equivalently, the improved Laplace approximation can be written as $I^{\texttt{iL}} = \hat{c}I^{\texttt{L}}$, with $\hat{c} = \prod_{i = 1}^{d}\hat{c}_i$ and $\hat{c}_i$ ($1\leq i \leq d$) defined in Scheme~\ref{alg:iLaplace}.

\begin{paragraph}{Remark 1} The factor $\hat c$ is an index of the magnitude of the improvement of the proposed method over the standard Laplace approximation. Indeed, values of $\hat{c}$ close to 1 indicate that the improved Laplace approximation is not improving over the standard Laplace. In this case, it is likely that the integrand is Gaussian-like. On the other hand, $\hat{c}\neq 1$ indicates that the integrand may not be Gaussian-like, e.g. it may be skewed and/or heavy-tailed. Note also that if $I$ factors as the product of $d$ scalar integrals, then the improved Laplace approximation of $I$ corresponds to its computation via numerical integration. 
\end{paragraph}

\begin{paragraph}{Remark 2}
An important difference of the proposed method from the integrated nested Laplace approximation (INLA) of \cite{rue2009approximate} is that our method is specifically designed to approximate normalising constants or marginal likelihoods for general models, and for arbitrary components of the parameter for both Bayesian and frequentist inference. For instance, the method can be used to approximate \eqref{eq:mlGLMM} even when random effects are not necessarily Gaussian. On the other hand, INLA is designed for approximating marginal posterior distributions and can approximate only Bayesian marginal likelihoods of Gaussian latent fields \citep[see, Eq. (30) of][]{rue2009approximate}. %In addition, the approximation of the marginal likelihood with INLA requires multivariate integration whereas the improved Laplace breaks the multivariate integration problem in $d$ univariate integrations, by sequentially applying the normalised Laplace approximation for scalar marginal densities. 
Some numerical comparison with INLA are provided in Section~\ref{sec:inla}. 
\end{paragraph}

\begin{paragraph}{Remark 3}
The order $(x_1,\ldots,x_d)$ is arbitrary, that is, the asymptotic error of the improved Laplace is not affected by their permutation. In practice, however, it may be useful to order $x$ according to the cardinality of the arguments of each element of $\nabla h(x)$, the gradient of $h(\cdot)$. In particular, the element of $x$ for which the corresponding element of $\nabla h(x)$ depends on all elements of $x$ may be placed as the first factor in \eqref{eq:expanded}. For instance, if $d=3$ and $\partial h(x)/\partial x_1$ depends on $x_1$, $\partial h(x)/\partial x_2$ depends on $x_{1:2}$ and $\partial h(x)/\partial x_3$ depends on $x$, then the order $(x_3,x_2,x_1)$, with $x_3$ being the first factor in \eqref{eq:expanded}, can simplify the conditional minimisations required by the proposed method. Obviously, when each element of $\nabla h(x)$ depends on $x_{1:3}$, there is no preferred ordering. In the examples considered in Section 3 we did not experience any practical difference in the results across different permutations of $x$.
\end{paragraph}

\begin{paragraph}{Remark 4}
Conditional minimisations can become computationally demanding when $d$ is large. Nevertheless, it is possible to avoid them by considering a first-order Taylor series expansion of the conditional minima as in \cite{cox1990approximation}. In particular, let $x =(y, z)$, where $y$ is the fixed block and $z$ the remaining part of $x$. Then $\hat{z}_{y}$, the conditional minimum of $h(\cdot)$ for fixed $y$, can be approximated by the linear regression
\begin{equation}
\tilde{z}_{y} = \hat{z} + V_{zz}^{-1}V_{zy}(\hat{y}-y)\,,\label{eq:appmode}
\end{equation}
which is such that $\hat{z}_{y} -\tilde{z}_{y} = O(n^{-1})$. Recently, \cite{kharroubi2016exponential} applied a similar idea in a different context and noted excellent performance (see also \citealp{ruli2014marginal}). We explore the numerical performance of the improved Laplace approximation with approximate conditional minima $\tilde{z}_y$ in place of their exact version $\hat{z}_y$ in Examples~\ref{sec:multT}, \ref{sec:inla}, \ref{sec:salam} and \ref{sec:rongelap}.
\end{paragraph}

\begin{paragraph}{Remark 5}
Linear constraints $Ax = b$, with $A$ and $b$ being appropriate matrices and vectors respectively, can be handled through a change of variables problem and by applying the proposed method to the remaining free components of $x$. Finally, when the distribution of $x_1|x_2$ is more Guassian-like and $x_2$ is low-dimensional, as it happens in the INLA framework, then it may be more sensible to approximate $x_1|x_2$ by the proposed method and integrate out $x_2$ by numerical integration. 
\end{paragraph}

\section{Examples}
\label{sec:examples}
The improved Laplace approximation is implemented in the \texttt{R} \citep{R-cran} package \texttt{iLaplace} \citep{ilaplace}, available on the \texttt{CRAN} repository.
%All computations are run on an Intel(R) Xeon(R) CPU E5-2620 2.00GHz CentOS machine. 
Except for the example of Section~\ref{sec:gompertz}, computations with the improved Laplace approximation, with either exact or approximate conditional minima, are performed in parallel over 11 threads through the parallel implementation provided in the package \texttt{iLaplace}. Essentially, it is an embarrassingly parallel implementation in which each integral in Steps 1-3 of Scheme~\ref{alg:iLaplace} is computed through a separate thread.
  
%In this section we illustrate the accuracy of the proposed method by three applications. The first and the second applications focus on Bayesian inference in nonlinear regression models, where the aim is to choose between the normal and the Student's $t$ error distributions. The third application is a logistic regression model with crossed random effects applied to the well known Salamander data \citep[see][p. 440]{mccullagh1989glm}. The Maximum Likelihood Estimate (MLE), obtained by maximising the approximate marginal log-likelihood computed by the improved Laplace method, is compared with the analogous obtained with other competing methods.
% {\em As \texttt{integrate} often fails to work with infinite domain, we set the domain of integration to $\hat{x} \pm 15 s$, where $s$ is a vector of standard deviations with $s_i$ being the square root of the first diagonal element of $(\hat{V}_{i:p,i:p})^{-1}$, $i=1,\ldots,p$. This choice guarantees that the domain of integration is wide enough to include the bulk of the integrand.}

\subsection{Gompertz distribution: fixed $d$ and $n\to\infty$}
\label{sec:gompertz}
Consider the sequence of sample sizes $\{n_i = \ceil{n_{i-1} + 1.2\sqrt{n_{i-1}}}, n_1 = 20, \text{and}\;\, i=2,\ldots,30\}$, where the symbol $\ceil{\cdot}$ denotes the ceiling function. Let $y^i = (y_1,\ldots,y_{n_i})$ be a random sample of size $n_i$ ($i=2,\ldots,30$) from the Gompertz distribution, with density $$p(y;\theta) = \alpha\beta e^{\beta y}e^{\alpha}\exp\{-\alpha e^{\beta y}\},\,$$ with $\theta = (\theta_1,\theta_2) = (\log\alpha,\log\beta)\in\Real^2$ and $y>0$. Moreover, for each $i$, consider 100 random datasets of size $n_i$ from the Gompertz distribution with $\alpha=2$ and $\beta = 3$. For each of these datasets, we compute the normalising constant of the posterior distribution $\pi(\theta|y^i) \propto L(\theta; y^i)\pi(\theta)$, where $L(\theta; y^i)$ is the likelihood function for $\theta$ based on data $y^i$ and $\pi(\theta) = \pi(\theta_1)\pi(\theta_2)$ is the prior distribution with $\pi(\theta_1)$ and $\pi(\theta_2)$ both being $N(0,100)$. 

The aim is to compare the behaviour of the standard ($\hat{I}^{\,\text{L}}$) and the improved ($\hat{I}^{\,\text{iL}}$) Laplace approximations with the target value ($I$) computed by adaptive numerical integration, as the sample size $n$ diverges. Similarly to \cite{diciccio2008conditional} and \cite{davison2006improved}, let $a_1>0$, $a_2>0$, $(b_1,b_2)\in\Real^2$ and suppose that
$$
I = \hat{I}^{\,\text{iL}}(1 + b_1 n^{-a_1}) + o(n^{-a_1}),
$$
and that
$$
I = \hat{I}^{\,\text{L}}(1 + b_2 n^{-a_2}) + o(n^{-a_2}),
$$
 for $n\to\infty$. Then, $\lim\limits_{n\to\infty}\{\hat{I}^{\,\text{iL}}/I\} = 1$ and $\lim\limits_{n\to\infty}\{\hat{I}^{\,\text{L}}/I\} = 1$, and if the improved Laplace approximation is more accurate than the standard Laplace, then the first limit should converge faster. {Furthermore, a log-log graph of $|I/\hat{I}^{\,\text{iL}}-1|$  ($|I/\hat{I}^{\,\text{L}}-1|$) against $n$ should be be linear with slope $-a_1$ ($-a_2$) and intercept $\log|b_1|$ ($\log|b_2|$)}.
 
The left panel of Figure~\ref{fig:1} reports the log-log plot of $\hat{I}^{\,\text{iL}}/I$ and $\hat{I}^{\,\text{L}}/I$, both averaged across the 100 repetitions at each value of and against the sample size $n$. This plot highlights that the improved Laplace approximation is more accurate than the standard Laplace method, since the visual convergence at 1 of $\hat{I}^{\,\text{iL}}/I$ happens at a faster rate.

\begin{figure}
\centering
\includegraphics[scale = 0.63, angle=-90]{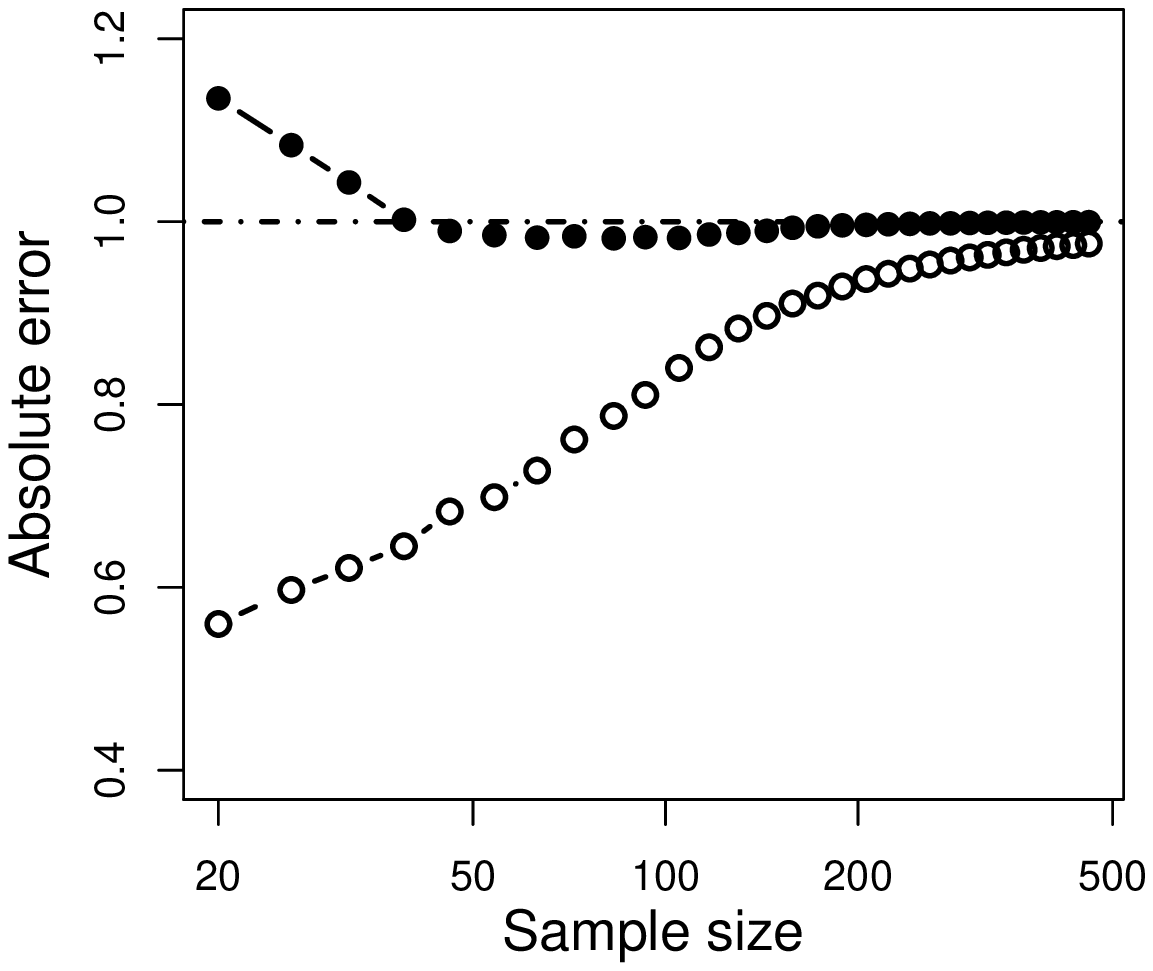}
\includegraphics[scale = 0.63, angle=-90]{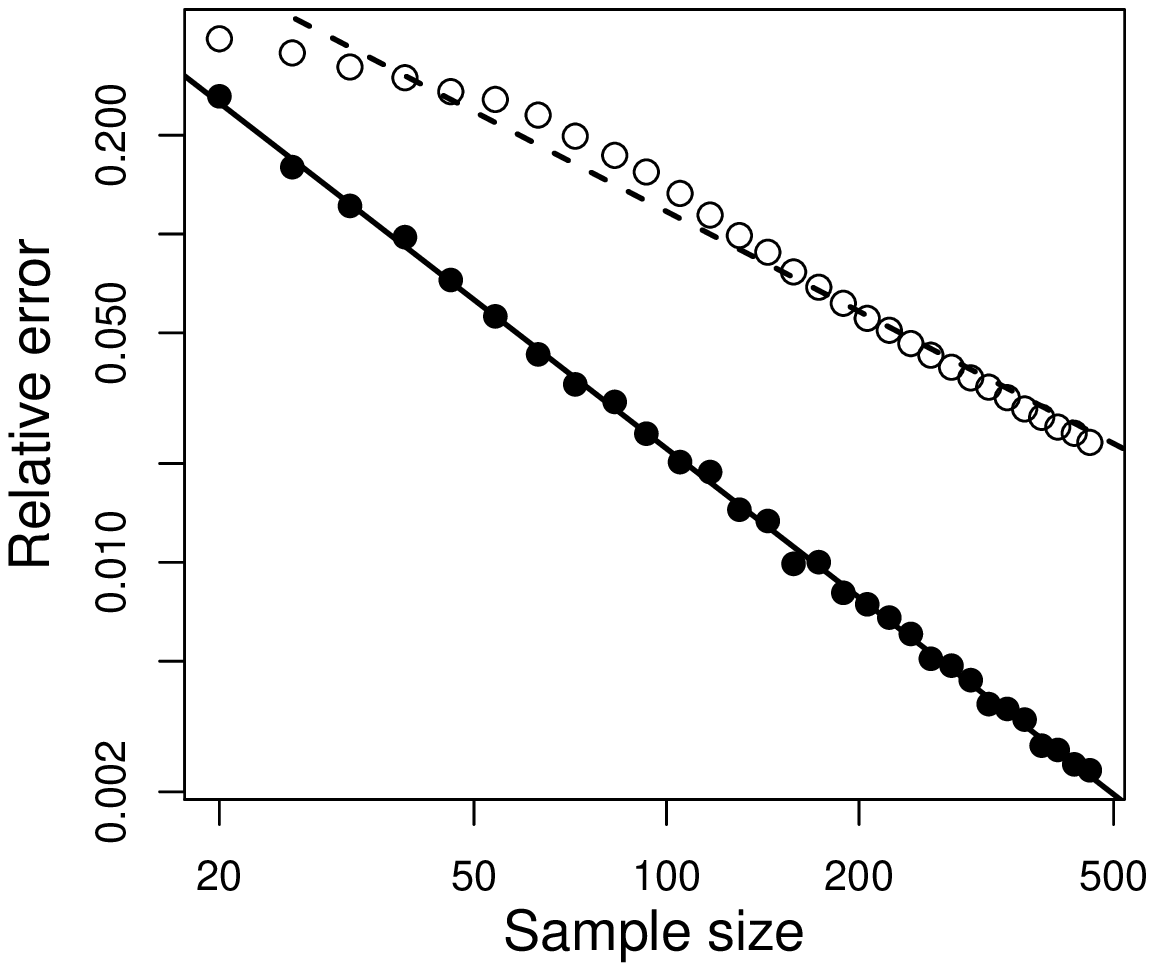}
\caption{Numerical evaluation of the asymptotic error of the improved and standard Laplace methods. Left panel: log-log plot of $(\hat{I}^{\,\text{iL}}/I)$ ($\bullet$) and $ (\hat{I}^{\,\text{L}}/I)$ ($\circ$) versus $n$. Right panel: log-log plot of the relative error $|\hat{I}^{\,\text{iL}}/I -1|$ ($\bullet$)  and $|\hat{I}^{\,\text{L}}/I -1|$ ($\circ$)  versus $n$; the solid and dashed lines are the corresponding least-squares regression lines. The empirical slope is -1.51, with 0.99 confidence interval (-1.53, -1.48), for the solid line, and -1.01, with 0.99 confidence interval (-1.09, -0.93), for the dashed line.}
\label{fig:1}
\end{figure}

The log-log plot of the relative error averaged across the 100 repetitions at each value of the sample size $n$ is shown on the right panel of Figure~\ref{fig:1}. This shows that the improved Laplace method achieves third-order accuracy whereas the standard Laplace formula is second-order accurate, e.g.  $a_1=1.5$ and $a_2=1$.

\subsection{Multivariate $t$/skew-$t$}
\label{sec:multT}
To assess the accuracy of the proposed method even in extreme settings, we consider the multivariate $t$/skew-$t$ distribution of \cite{jones2002marginal}, with density 

\begin{eqnarray}\nonumber
p(y;\nu,a,c)&=&\frac{\Gamma((\nu+d)/2)}{\Gamma((\nu+1)/2)B(a,c)(a+c)^{1/2}2^{a+c-1}(\nu\pi)^{(d-1)/2}}\\\nonumber
&& \times \frac{(1+\nu^{-1}y_1^2)^{(v+1)/2}\left(1+\frac{y_1}{(a+c+y_1^2)^{1/2}}\right)^{a+1/2}\left(1-\frac{y_1}{(a+c+y_1^2)^{1/2}}\right)^{c+1/2}}{\{1+\nu^{-1}(y_1^2+\ldots+y_d^2)\}^{(\nu+d)/2}}\,.
\end{eqnarray} 
The positive parameters $a$ and $c$ determine the distribution of the skewed marginal, the parameter $\nu$, i.e. the degrees of freedom (df), controls the tail behaviour of the distribution, and $B(\cdot,\cdot)$ and $\Gamma(\cdot)$ are the beta and gamma functions, respectively. This distribution is obtained from the multivariate Student's $t$-density centred at $0$ and with identity scale matrix, where the marginal density of the first component is replaced with the univariate $t$/skew-$t$ density. The case with $a = c = \nu/2$ leads the the ordinary multivariate Student's $t$-distribution with identity scale matrix and $\nu$ degrees of freedom.

We approximate the normalizing constant of the multivariate $t$/skew-$t$ density, with the standard and the improved Laplace approximations, in two scenarios: the first with $a = c = 1.5$, and the second with $a = 12$ and $c = 0.5$. For each scenario, we consider  multivariate $t$/skew-$t$ densities with varying dimension and degrees of freedom. 
Results in the first row of Figure~\ref{fig:skewt} show that the standard Laplace approximation rapidly deteriorates with increasing dimensionality, and increasing non-Gaussianity, i.e. low $\nu$, higher skewness (large $a$ and small $c$ or vice versa). On the contrary, the improved Laplace approximation is reasonably accurate and stable across both increasing dimensionality and non-Gaussianity. A similar example has been considered also by \cite{nott2009generalization}, in order to test the accuracy of their modified Laplace approximation. However, their method is substantially less accurate then ours, only slightly improving the poor quality of the standard Laplace approximation. For instance, the normalising constant of the 10-variate $t$/skew-$t$ density with $a = 4$, $c = 1$ and $\nu = 3$, is 0.013 with the standard Laplace, 0.02 with the modified Laplace method of \cite{nott2009generalization} and 0.9981 with the improved Laplace approximation.

\begin{figure}
\centering
%\subfigure[Impvored vs standard Laplace approximations]{
\includegraphics[scale = 0.62, angle=-90]{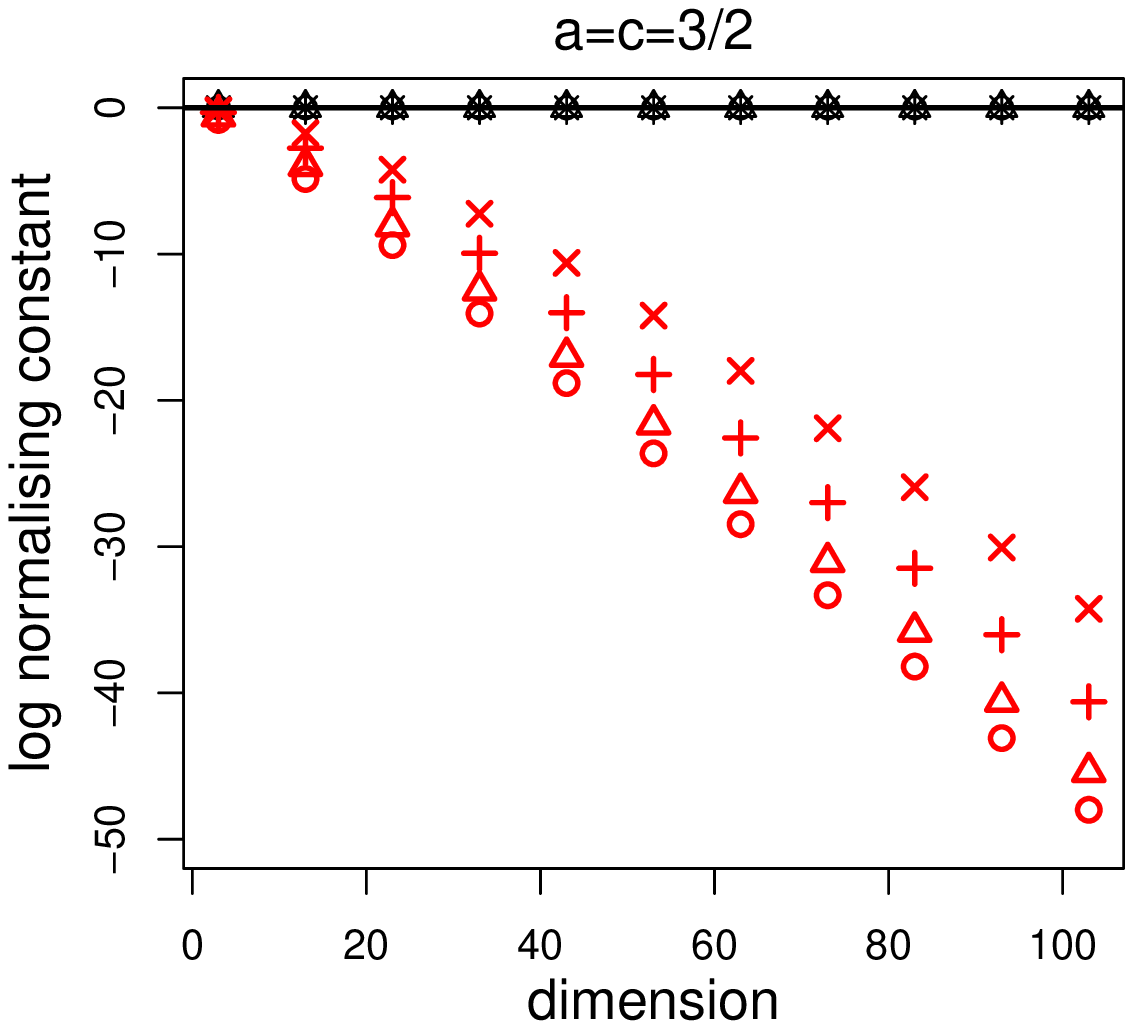}
\includegraphics[scale = 0.62, angle=-90]{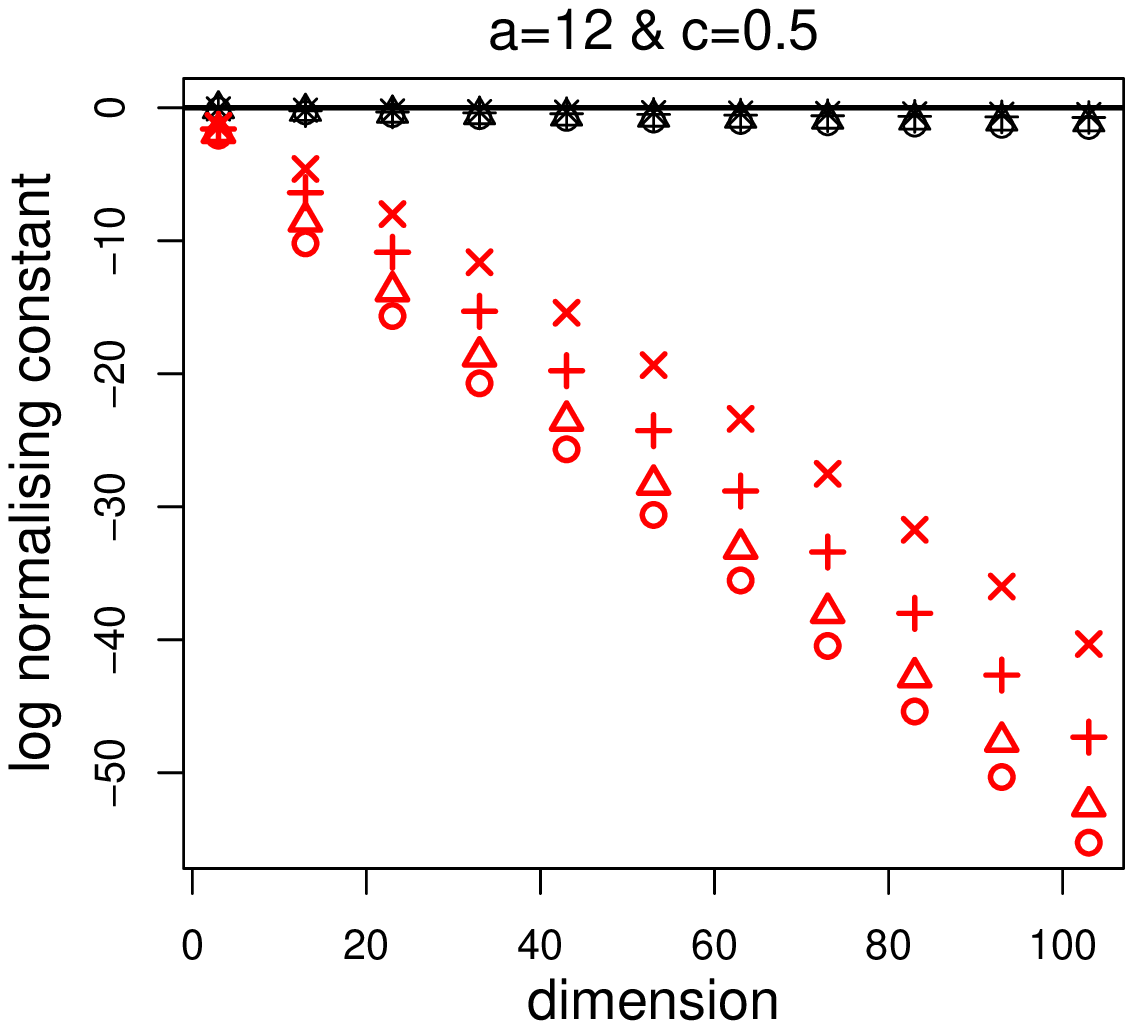}\\
%}
%\subfigure[Improved Laplace approximation with exact and approximate conditional optimisations]{
\includegraphics[scale = 0.62, angle=-90]{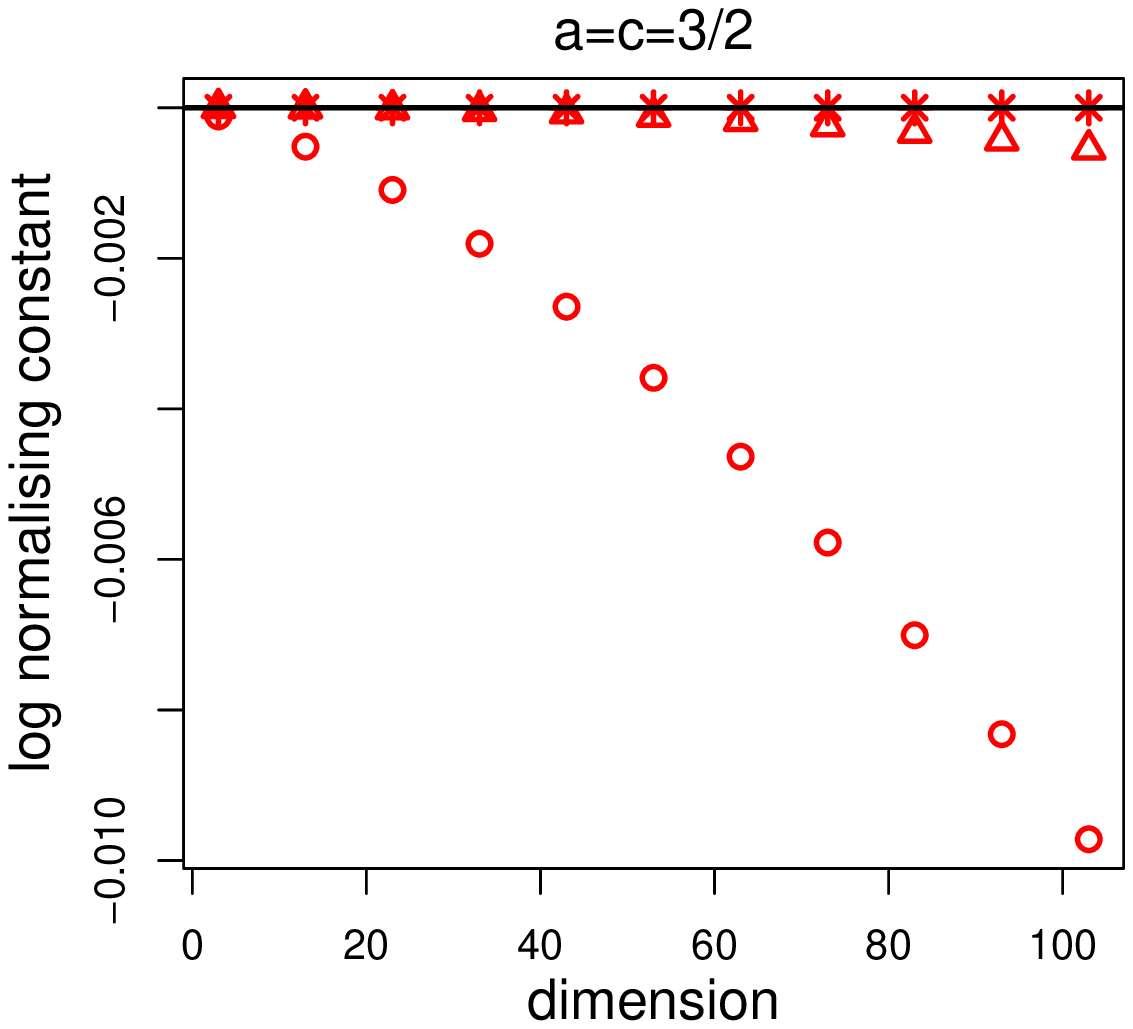}
\includegraphics[scale = 0.62, angle=-90]{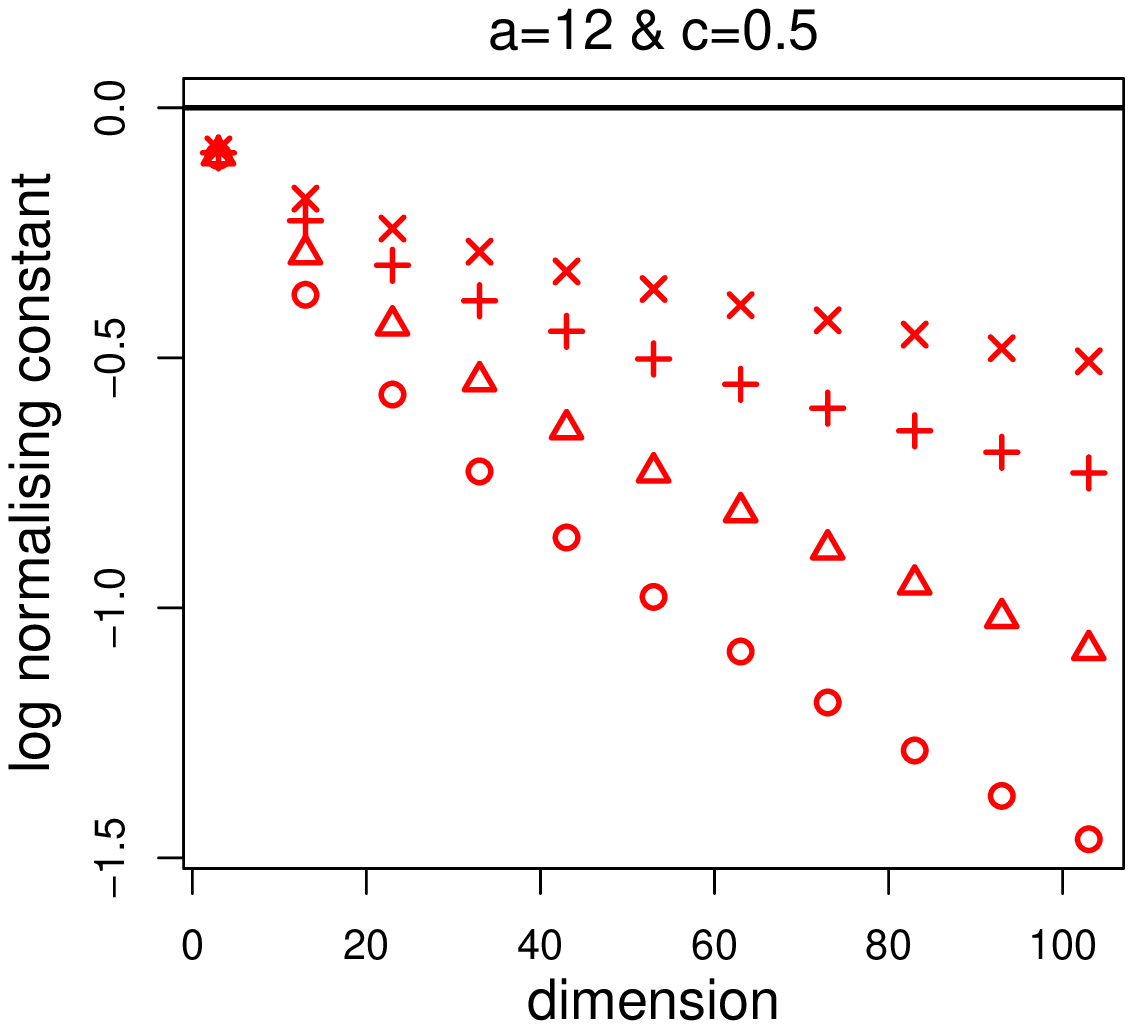}
%}
\caption{First row: standard (red coloured symbols) and improved Laplace (black) approximations of the normalizing constant (in logarithmic scale) of the multivariate $t$/skew-$t$ density against dimension.  Second row: improved Laplace approximation with either exact or approximate conditional minima overlap. Degrees of freedom are: 3 ({\large$\circ$}), 5 ({\smaller$\triangle$}), 10 ($+$) and 20 ($\times$).}\label{fig:skewt}
\end{figure}

Now consider the same example but using the approximate conditional minima introduced in Remark 4 of Section 2.1. As shown in the second row of Figure~\ref{fig:skewt}, in this case, the results of the improved Laplace approximation with either actual or approximate conditional minima coincide.

\subsection{A comparison with INLA}
\label{sec:inla}
The following example has been considered by \cite{ferkingstad2015improving} and is known as challenging for the INLA methodology. %In addition, here also the Laplace approximation can be dramatically inaccurate in approximating the posterior normalising constant. 

Let $y=(y_1,\ldots,y_n)$ be conditionally independent binary values with 
\[
Y_i|u_i\sim\text{Bernoulli}(p_i)\,,
\]
\[
\text{logit}(p_i) = \beta + u_i\,,
\] 
where $U_i\sim N(0,\sigma^2)$, for $i=1,\ldots,n$. We assume independent priors for $\beta$ and $\nu = \sigma^{-2}$, with $\beta\sim N(0,1)$ and $\nu \sim\text{Gamma}(1,1)$ and we wish to obtain the marginal posterior distributions of $\beta$ and $\nu$. % and to compute the posterior normalising constant.

We apply the improved and the standard Laplace methods to approximate $L(\beta,\nu;y)$, the marginal likelihood defined in \eqref{eq:mlGLMM} under the aforementioned modelling assumptions. Since $L(\beta,\nu;y)\pi(\beta,\nu)$ is bivariate, we use adaptive numerical integration for obtaining the marginal posteriors $\pi(\beta|y)$ and $\pi(\nu|y)$. For comparison purposes, the marginal posteriors are also approximated by: MCMC, the standard version of INLA and the improved INLA proposed by \cite{ferkingstad2015improving}.  For the MCMC approximation, we consider $10^6$ final posterior samples with the \texttt{JAGS} software \citep{plummer}, after a burn-in of $10^6$ samples. Computations with INLA are done using the associated \texttt{R} package \texttt{R-INLA}, using the default options.

%To approximate the posterior normalising constant we use either the improved or the standard Laplace to integrate out $(\beta, u)$ and numerical integration for integrating over $\log\sigma^{-2}$. These approximations are similar in spirit to the INLA approximation \citep[see Eq. (30) of ][]{rue2009approximate} as implemented in the \texttt{INLA} package.

Note that the integral over $u = (u_1,\ldots,u_n)$, required for obtaining $L(\beta, \nu;y)$, can be factorised as product of $n$ scalar integrals. Hence, in this case the improved Laplace method is as accurate in approximating $\pi(\beta,\nu|y)$ as is numerical integration.

As an illustration, consider a sample of size $n=100$ generated from the model with $\sigma^2=1$ and $\beta=2$.
\begin{figure}
\includegraphics[scale = 0.63, angle=-90]{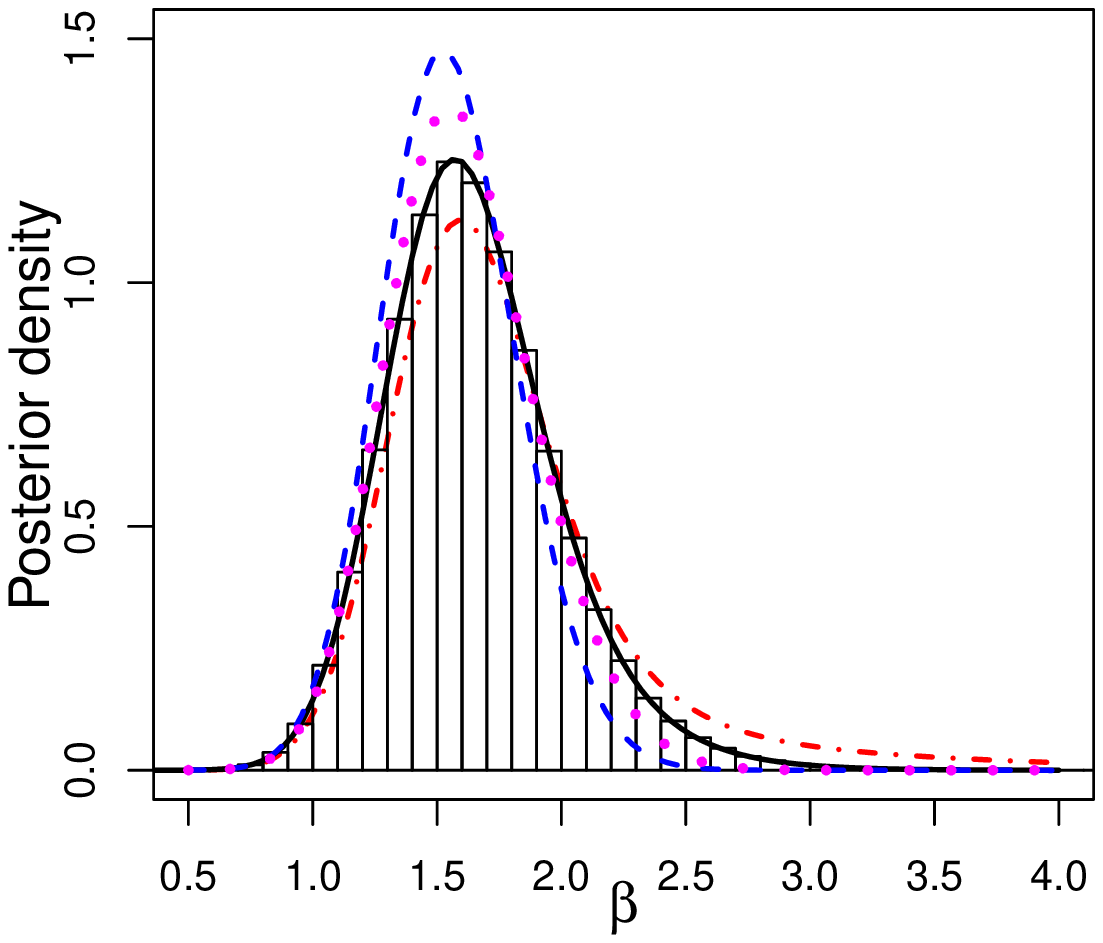}
\includegraphics[scale = 0.63, angle=-90]{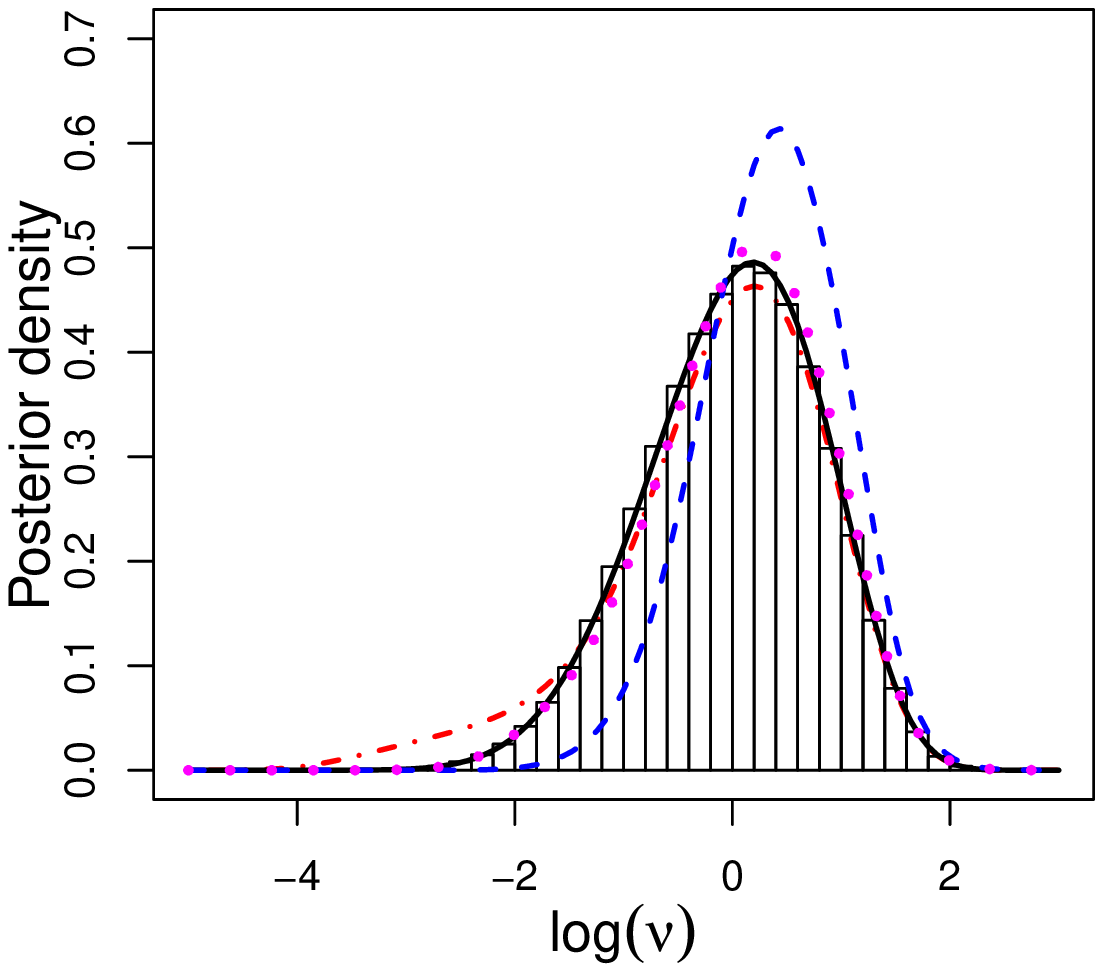}\\
\includegraphics[scale = 0.63, angle=-90]{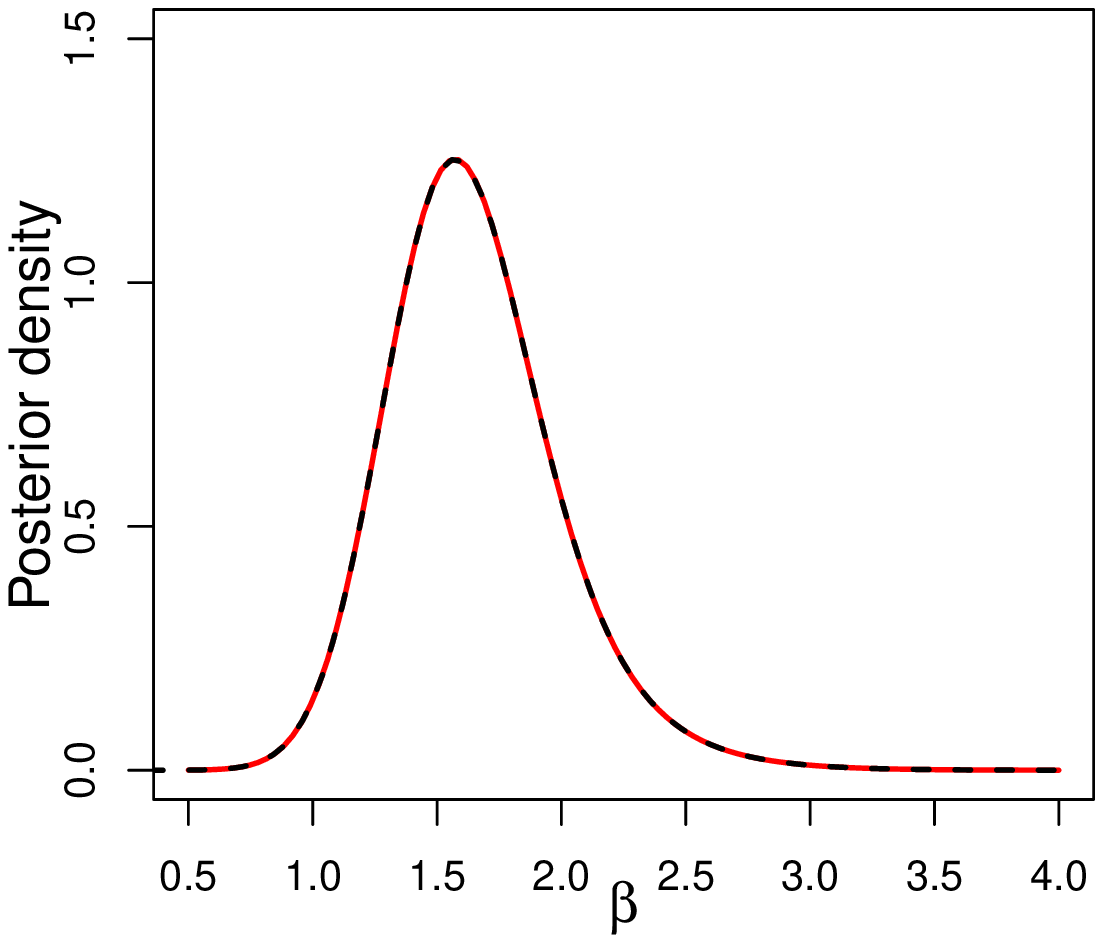}
\includegraphics[scale = 0.63, angle=-90]{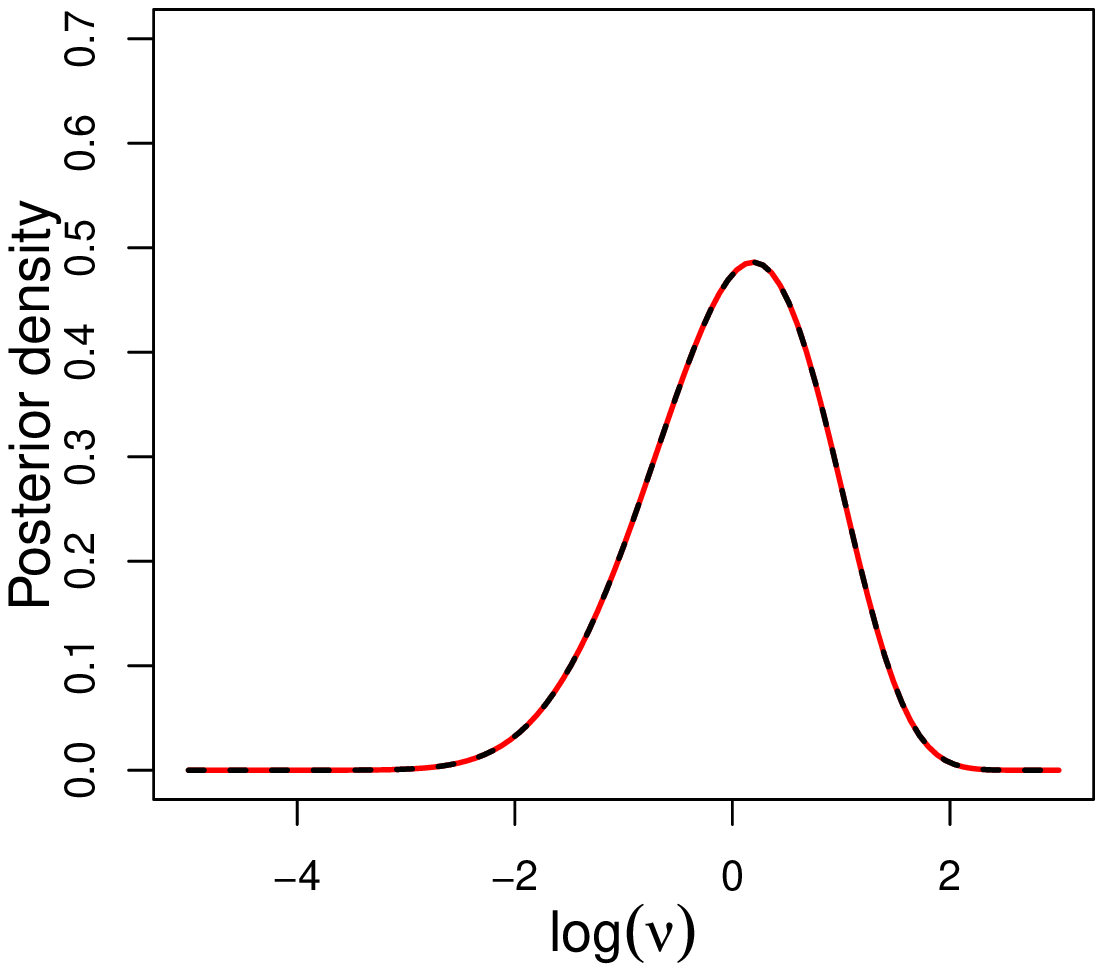}
\caption{First row: marginal posterior distributions for $\beta$ and $\log \nu$ approximated by: $10^6$ MCMC samples with \texttt{JAGS} (histogram), Laplace (dash-dotted red curves), improved Laplace (solid black), INLA (dashed blue) and improved INLA (dotted magenta) methods. Second row: marginal posteriors of $\beta$ and $\log \nu$ with the improved Laplace approximation with exact (black) and approximate conditional minima (red dashed).}
\label{fig:figINLA}
\end{figure}
As a gold standard we use MCMC as implemented in the \texttt{JAGS} software. We compare it with the improved and the standard Laplace approximations and with INLA and the improved ILNA of \cite{ferkingstad2015improving} in the first row of Figure~\ref{fig:figINLA}. As expected, the improved Laplace approximation is virtually indistinguishable from the MCMC approximation. On the other hand, the Laplace approximation and both versions of INLA perform slightly worse then the improved Laplace approximation.

The second row of Figure~\ref{fig:figINLA} shows the marginal posteriors of $\beta$ and $\log\nu$ approximated by the improved Laplace method using the approximate conditional minima introduced in Remark 4. In this case, the improved Laplace approximation with either exact (black continued) or approximate (red dashed) conditional minima gives indistinguishable results.

%Perhaps, more interesting is the approximation of the posterior normalising constant. With numerical integration we obtained (in logarithmic scale) -51.54. In the first strategy, the Laplace and the improved Laplace methods give -52.05 and -51.40, respectively. The INLA and the improved INLA methods give respectively, -52.06 and -51.55. In the second strategy with the Laplace and the improved Laplace we obtain respectively, xx and xx.  with standard Laplace and with the improved Laplace $-51.54$. In the second strategy the standard Laplace method yields xx whereas the improved Laplace method gives xx.
%
%The high inaccuracy of the Laplace method in the second strategy is due to the fact that the global minimum of $h(u,\theta)$ is quite different from the modes of the marginal densities of marginals.
%%%%%%%%%%%%%%%%%%%%%%%%%%%%%%%%%%%%%%%%%%%%%%%%%%%%%%%%%%%%%%%

A small simulation study is performed in order to assess the accuracy of the proposed method. In particular, 
we consider 100 datasets with sample size $n = 100$ drawn from the model and under the same parameter values as before. For each dataset, we compute the marginal posteriors of $\theta$ by MCMC, here treated again as the gold standard, and by: the standard Laplace, the improved Lapalce (with approximate conditional minima) and the original and corrected versions of INLA.  The MCMC approximation is done by $10^6$ samples, after a burn-in of $10^5$, and thinning equal to 10. As a measure of discrepancy, we compute the Kullback-Leibler (KL) divergence between the MCMC posterior and the other approximation methods. The KL divergence is defined as

\[
KL(\pi;\tilde\pi) = \int_{\theta\in\Theta}\log\left\{\frac{\pi(\theta|y)}{\tilde\pi(\theta|y)}\right\}\pi(\theta|y)\,\dd\theta
\,,
\]
where $\pi(\theta|y)$ is the MCMC posterior and  $\tilde\pi(\theta|y)$ is the approximate posterior obtained with the other methods. For simplicity, we compute two marginal KL divergences, i.e. one for $\beta$ and one for $\nu$. The higher the KL, the  worse is the approximation $\tilde\pi(\cdot|y)$. The MCMC marginal posteriors are computed with logspline density estimation using the \texttt{logspline} package of \texttt{R}. This tends to give smoother density estimates than usual kernel density estimators. The marginal posterior distributions obtained with either standard or improved Laplace approximation are available analytically, whereas those based on the two versions of INLA are build through smoothing splines.

\begin{figure}
\includegraphics[scale = 0.63, angle=-90]{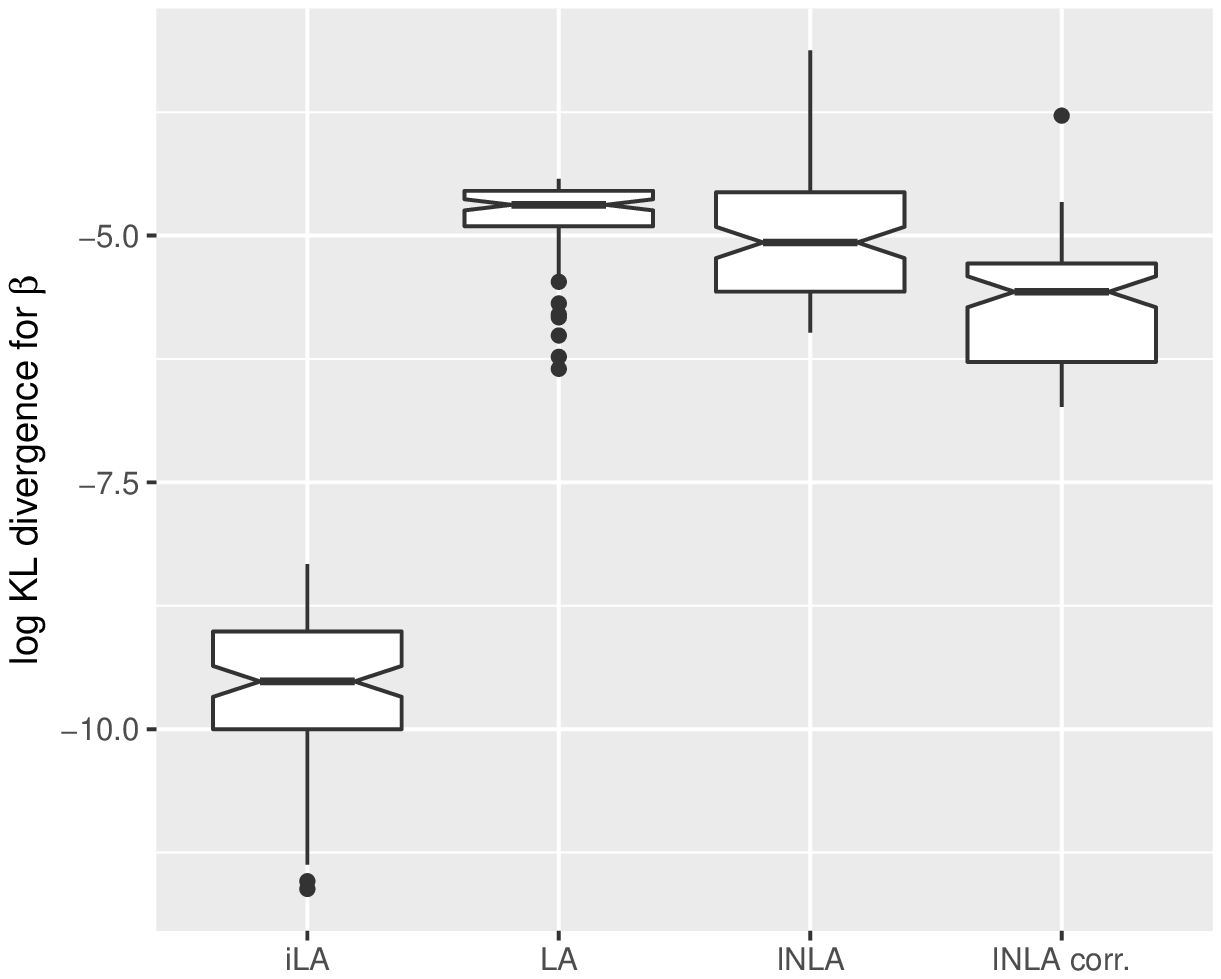}
\includegraphics[scale = 0.63, angle=-90]{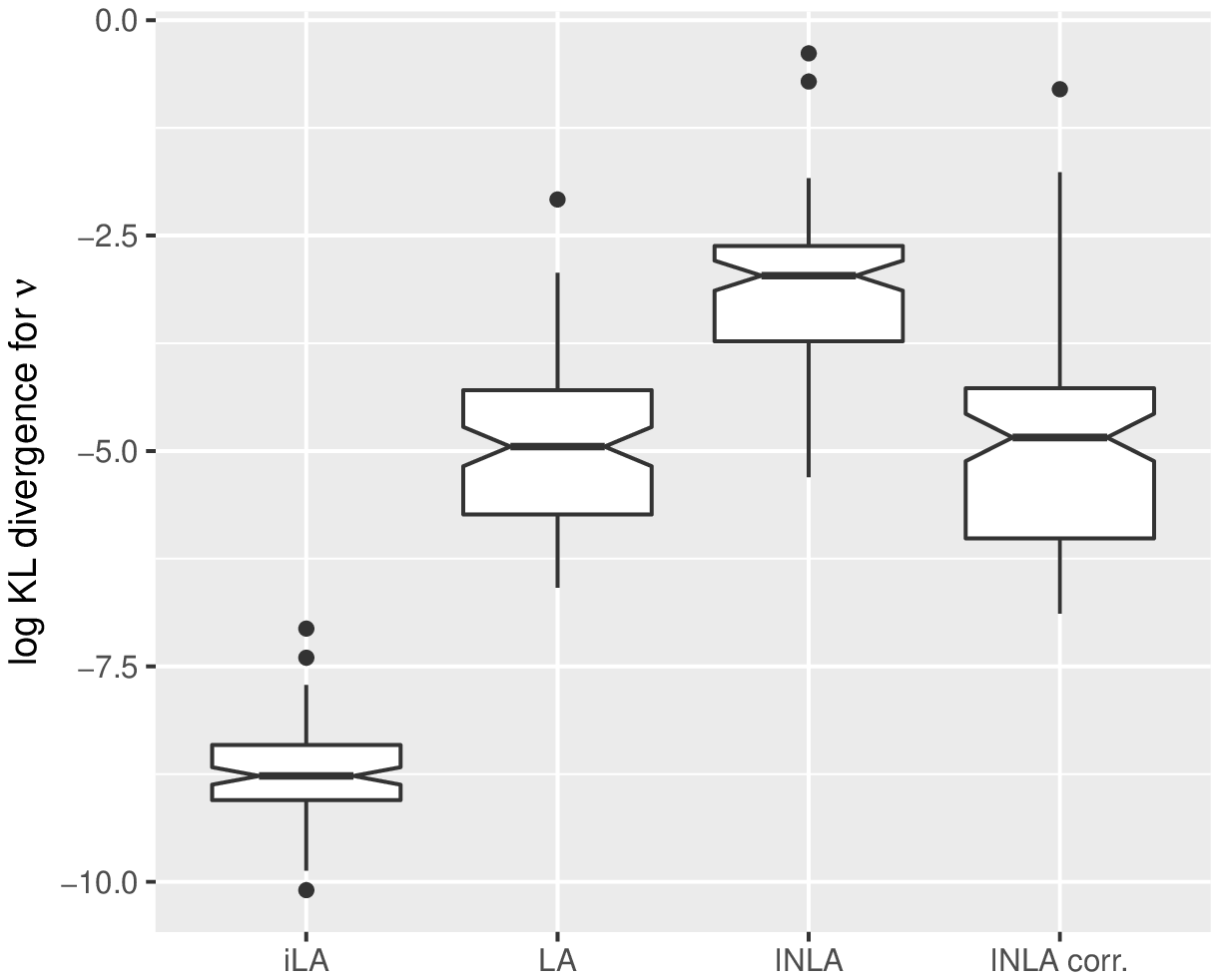}\\
\caption{Log-KL divergences of marginal MCMC posteriors of $\beta$ (left panel) and $\nu$ (right panel) from the corresponding ones approximated through the: Laplace (LA), improved Laplace (iLA), INLA and the corrected INLA methods in a repeated sampling context with 100 replications.}
\label{fig:figINLA2}
\end{figure}

The results in Figure~\ref{fig:figINLA2} highlight that the marginal posteriors of $\beta$ (left panel) and $\nu$ (right panel) approximated with the proposed method are the closest to the MCMC posteriors in terms of the KL divergence.

\subsection{Nonlinear regression}
\label{sec:nonlin}
Consider the nonlinear regression model
\[
y_i = \beta_1\exp\{1-\exp(-x_i/\beta_2)\} + \sigma\varepsilon_i,
\]
where $y_i$ is the response variable, $x_i$ is a covariate, $(\beta_1,\beta_2)$ are unknown regression parameters, and the $\epsilon_i$ are independent error terms, $i = 1, \ldots,n$. We focus on two possible distributions for the error term: the normal distribution and the Student's $t$-distribution, with unknown degrees of freedom $\tau$. The aim is to choose among them through the Bayes Factor (BF), which in our case is given by the ratio of the posterior normalising constant of the normal model over that of the Student's $t$ model.

As an example we consider the BOD2 dataset \citep[][p. 305]{bates1988nonlinear}, which concerns a study on biochemical oxygen concentration ($y$) as function of time ($x$). For both models, we assume the parameters are a priori independent. Moreover, a bivariate normal distribution with mean vector zero and scale matrix $10\mathbf{I}_2$, is assumed for $\beta$. Following the recommendations of \cite{gelman2006prior}, for the scale parameter $\sigma$ we assume a half-Cauchy prior with scale equal to 10. Finally, the Jeffreys rule prior proposed by \cite{fonseca2008objective} is taken for $\tau$. For numerical stability in the optimisations, $\sigma$ and $\tau$ are considered in logarithmic scale.

Table~\ref{tab:nlinreg} shows the log-marginal likelihoods and the Bayes factor approximated by the improved Laplace approximation and by: the standard Laplace, the Bartlett-corrected Laplace (\citealt{diciccio1997computing}), importance sampling, the method of \cite{chib2001} and adaptive numerical integration as implemented in the \texttt{R} package \texttt{cubature}. The \citeauthor{chib2001}'s, IS and Bartlett-corrected Laplace approximations are replicated 500 times, where for each replication the MCMC algorithm is started at a different point. The final estimates of the log-marginal likelihood and of the BF (in decimal logarithmic scale) are obtained by averaging the 500 replications. At each replication, the Bartlett-corrected Laplace approximation is performed with $2\times10^5$ final MCMC posterior draws, after suitable burn-in and thinning to reduce autocorrelation. The same MCMC posterior sample is used also for computing the marginal likelihood with \citeauthor{chib2001}'s method. For the IS approximation we consider $10^6$ draws from the multivariate Student's $t$-distribution with $m$ degrees of freedom, centred at the posterior mode, with scale matrix equal to $c>0$ times the inverse of the posterior Hessian at the modal value. Several values of $c$ around 1 and $m\in[3,50]$ were also considered. However, they gave very similar results, so we set $c = 1$ and $m = 3$. This choice permits to have an importance density with finite variance and with heavy tails (see, e.g., \citealp{evans2000}, Sect. 6.3). The standard deviation of the 500 log-marginal likelihoods divided by $\sqrt{500}$ is taken as a measure of Monte Carlo standard error.

\begin{table}[h!]
\centering\smaller
\begin{tabular}{ll c c c c c}
\hline
\hline
 Model & Adaptive    &Laplace & Improved & Bartlett-   & Chib \& & IS\\
       & integration &          & Laplace  & corrected (3$\times$SE$_{\text{MC}}$)  &   Jeliazkov (3$\times$SE$_{\text{MC}}$)   & (3$\times$SE$_{\text{MC}}$)\\
\hline
Normal   & -2.539  & -2.905 & -2.540 & -2.504 (0.0198) & -2.537 (0.003) & -2.539 (0.0001) \\
Student's $t$ & -2.488 & -5.179 & -2.449 & -4.188 (0.0199) & -2.478 (0.0031) & -2.457 (0.0001)\\
\hline
$\log_{10}$ BF & -0.022 & 0.988 & -0.039 & 0.731 (0.0103) & -0.027 (0.0057) & -0.036 (0.0183)\\
\multicolumn{2}{l}{(Normal vs $t$)}&&&&&\\
\hline
\hline
\end{tabular}
\caption{BOD2 data. Logarithm of Bayesian marginal likelihoods and BF computed with: adaptive numerical integration, standard Laplace, the improved Laplace, the Bartlett-corrected Laplace, importance sampling (IS) and \citeauthor{chib2001}'s method. For the last three methods, the point estimates are obtained by averaging over 500 replications; the quantity in parenthesis gives $3\times$SE$_{\text{MC}}$, where SE$_{\text{MC}}$ is the Monte Carlo standard error given by the standard deviation divided by $\sqrt{500}$.}\label{tab:nlinreg}
\end{table}

Results in Table~\ref{tab:nlinreg} indicate that the standard Laplace approximation and its Bartlett-corrected version are quite inaccurate, since both lead to substantial evidence in favor of the normal model (see \citealp{kass1995bayes} for the interpretation of the BF). Such an evidence is not confirmed by the BF approximated through adaptive numerical integration, here treated as the gold standard; neither IS and \citeauthor{chib2001}'s method confirm the aforementioned evidence. In addition, results of the improved Laplace method are in reasonable agreement with IS, \citeauthor{chib2001}'s approximation and adaptive numerical integration.

The inaccuracy of the standard Laplace approximation in the case of the Student's $t$ model is most likely due to the non normality of the marginal posterior of $(\log\sigma,\log\nu)$. Indeed, a look at the bivariate kernel density estimate of this marginal bivariate posterior (not reported here) reveals that it is banana-shaped, and therefore it is quite far from being elliptical. Nevertheless, such a shape is well accommodated by the improved Laplace approximation.

\subsection{GLMM with crossed random effects}
\label{sec:salam}
We consider the problem of approximating the marginal likelihood for the fixed parameters in a model with crossed random effects \citep{shun1995laplace,shun1997another}. Such a model is useful, for instance, when analysing the Salamander mating data \citep[][p. 439]{mccullagh1989glm}. These data have been analysed by \cite{karim1992}, \cite{shun1997another}, \cite{Booth1999wt}, \cite{bellio2005pairwise}, \cite{sung2007}, among others, and consist of three separate experiments, each performed according to the design given in \citet[][Table 14.3]{mccullagh1989glm}. Each experiment involved matings among salamanders in two closed groups. Both groups contained five species {R} females, five species {W} females, five species {R} males and five species {W} males. Within each group, only 60 of the possible 100 heterosexual crosses were observed owing to time constraints. Thus, each experiment resulted in 120 binary observations indicating which matings were successful and which were not.

As in \citet[][p. 441]{mccullagh1989glm}, the data are modelled as if different sets of 20 male and 20 female salamanders were used in each experiment. Let $y_{ij}$ be the indicator of a successful mating between female $i$ and male $j$, for $i, j=1,\ldots,60$, where only 360 of the $(i,j)$ pairs are relevant. Let $u_i^f$ denote the random effect that the $i$th female salamander has across matings in which she is involved, and define $u_j^m$ similarly for the $j$th male. The data $y_{ij}$ are assumed conditionally independent with
\[
Y_{ij}|u_i^f,u_j^m\sim \mathrm{Bernoulli}(p_{ij}),
\]
\[
\text{logit}(p_{ij}) = x_{ij}^\T\beta + u_{i}^f + u_j^m,
\]
where $x_{ij}$ is a 4-dimensional row vector of zeros and ones indicating the type of cross, $\beta$ is the vector of fixed effects, $U_{i}^f\sim N(0,\sigma_f^2)$ and $U_{j}^m\sim N(0,\sigma_m^2)$. 

As a first example we consider the estimation of $\beta$ and $(\sigma_f^2, \sigma_m^2)$ for each separate experiment -- following the same model structure as in \cite{shun1997another} -- performed by maximising the approximate marginal likelihood. The aim is to compare the maximum likelihood estimate (MLE) based on the marginal likelihood approximated by the improved Laplace method with those based on the modified Laplace approximation proposed by \cite{shun1995laplace} and \cite{shun1997another}. It is well known that in models with crossed random effects the standard Laplace approximation is not asymptotically valid \citep{shun1995laplace}, and it may give poor results.

Let $\beta = (\beta_0, \beta_{\text{WS}_{f}}, \beta_{\text{WS}_{m}},\beta_{\text{WS}_{f}\times \text{WS}_{m}})$ be the fixed effects, where $\beta_{0}$ is a constant, $\beta_{\text{WS}_{f}}$ is the effect of the dummy variable $\text{WS}_{f}$ which takes one if the observation is from a species W female and zero otherwise, and so on. The marginal likelihood has the form \eqref{eq:mlGLMM}, with $\theta = \beta$ and $\theta_u = (\sigma_f^2, \sigma_m^2)$ and
%\[
%%L(\beta,\sigma_f^2,\sigma_m^2) = \int\int L(\beta,\sigma_f^2,\sigma_m^2;u_f,u_m)f(u_f;\sigma_f^2)f(u_m;\sigma_m^2)\,\dd u_f\dd u_m
%,\]
involves a 40 dimensional integral that cannot be reduced to a product of lower dimensional integrals, even though the random effects $u_f=(u^f_1,\ldots,u^f_{20})$ and $u_m=(u_1^m,\ldots,u_{20}^m)$ have independent normal distributions. The approximate MLE for the three separate experiments (reported in Tab~\ref{tab:salamander}) are compared with those of \citet[][Tab. 2 and Tab. 3]{shun1997another}.

\begin{table}
\small
 \begin{minipage}[t]{0.5\linewidth}
 \begin{tabular}{lcccccccc}
 \hline
 \hline
 &\multicolumn{6}{c}{Approximate MLE} &  & \\
 \cline{2-9}
Methods & $\beta_0$ & $\beta_{WS_{f}}$ & $\beta_{WS_{m}}$ & $\beta_{WS_{f}\times WS_{m}}$ & $\sigma_f^2$ & $\sigma_m^2$ & Sec. & N. of. Iter.\\
\hline
\multicolumn{7}{l}{Laplace:}\\
\hspace{0.5cm}Exper. 1 & 1.34 & -2.94 & -0.42 & 3.18 & 1.58 & 0.073 &  0.92 & 22\\
\hspace{0.5cm}Exper. 2 & 0.57 & -2.46 & -0.77 & 3.71 & 1.81 & 0.92 & 0.72 & 15 \\
\hspace{0.5cm}Exper. 3 & 1.02 & -3.23 & -0.82 & 3.82 & 0.35 & 1.85 &1.05 & 22 \\
\multicolumn{7}{l}{Modified Laplace of \cite{shun1997another}\footnote{Corrected(1) values taken from \cite{shun1997another}}:} & & \\
\hspace{0.5cm}Exper. 1 & 1.37 & -3.02 & -0.44 & 3.27 & 1.72 & 0.185 & $-$& $-$ \\
\hspace{0.5cm}Exper. 2 & 0.57 & -2.53 & -0.77 & 3.79 & 2.10 & 1.10 & $-$& $-$ \\
\hspace{0.5cm}Exper. 3 & 1.04 & -3.31 & -0.83 & 3.90 & 0.46 & 2.07 & $-$& $-$ \\
\multicolumn{7}{l}{Improved Laplace:} & & \\
\hspace{0.5cm}Exper. 1 & 1.37 & -3.02 & -0.44 & 3.27 & 1.74 & 0.189 & 397 & 29\\
\hspace{0.5cm}Exper. 2 & 0.56 & -2.55 & -0.79 & 3.77 & 2.12 & 1.14 & 209 & 17\\
\hspace{0.5cm}Exper. 3 & 1.03 & -3.30 & -0.82 & 3.90 & 0.49 & 2.12 &  145 & 11 \\
\multicolumn{7}{l}{Improved Laplace with approximate conditional minima:} & &\\
\hspace{0.5cm}Exper. 1 & 1.36 & -2.99 & -0.44 & 3.24 & 1.72 & 0.15 & 56 & 15\\
\hspace{0.5cm}Exper. 2 & 0.56 & -2.49 & -0.75 & 3.72 & 2.07 & 1.05 & 64 & 16\\
\hspace{0.5cm}Exper. 3 & 1.02 & -3.27 & -0.82 & 3.87 & 0.43 & 2.03 & 83 & 20\\
 \hline
 \hline
 \end{tabular}
 \end{minipage}
 \caption{Salamander data. Comparison of the improved Laplace method (with exact and approximate conditional minima) with the standard Laplace and the modified Laplace approximation of \cite{shun1997another}. ``Sec." refers to the elapsed time in seconds and ``N. of Iter." refers to the number of iterations before convergence.}
 \label{tab:salamander}
 \end{table}
 
 From Table~\ref{tab:salamander} we notice that the standard Laplace method is very fast but the resulting MLE are quite inaccurate as far as variance components parameters are concerned. On the other hand, approximate MLE obtained with the improved Laplace method, with either exact or approximate conditional minima, are closer to those based on the modified Laplace approximation of \cite{shun1997another}. However, the improved Laplace method is easier to compute since it does not require derivatives of the negative log-integrand beyond the second-order. In terms of computing time, the improved Laplace approximation with approximate conditional minima is much faster than the version with exact conditional minima, though slower than the standard Laplace approximation. 

Consider now the joint analysis of the Salamander data, by independently combining the three experiments' data. In this case the marginal likelihood entails the computation of three 40-dimensional integrals. To compare our method with other results available in the literature, we consider a slightly modified version of the fixed effects. In particular, here $\beta$ is equal to $(\beta_\text{R/R}, \beta_\text{R/W}, \beta_\text{W/R}, \beta_\text{W/W})$, where $\beta_\text{R/R}$ denotes the effect of the cross between a species R female and a species R male, and so on. 
 \begin{table}
  \begin{minipage}[t]{0.5\linewidth}
  \begin{tabular}{lcccccc}
  \hline
  \hline
  &\multicolumn{6}{c}{Approximate MLE}\\
  \cline{2-7}
 Methods & $\beta_0$ & $\beta_{WS_{f}}$ & $\beta_{WS_{m}}$ & $\beta_{WS_{f}\times WS_{m}}$ & $\sigma_f^2$ & $\sigma_m^2$ \\
 \hline
 Laplace & 1.01 & 0.31 & -1.90  & 0.99 & 1.17 & 1.04\\
 Improved Laplace     & 1.02 & 0.32 & -1.95  & 1.00 & 1.39 & 1.25\\
 Improved Laplace (approx. cond. min)     & 1.01 & 0.31 & -1.92  & 0.98 & 1.34 & 1.19\\
 MC-EM \citep{Booth1999wt} & 1.03 & 0.32 & -1.95 & 0.99 & 1.40 & 1.25\\
 Gibbs \citep{karim1992} & 1.03 & 0.34 & -1.98 & 1.07 & 1.50 & 1.36\\
 PQL\footnote{From \cite{Booth1999wt}}   & 0.87  & 0.28 & -1.69 & 0.95& 1.35 & 0.93\\
  \hline
  \hline
  \end{tabular}
  \end{minipage}
  \caption{Salamander mating data analysed jointly. Comparisons of the improved Laplace approximation (with either exact or approximate conditional minima) with the Monte Carlo Expectation-Maximisation (MC-EM) of \cite{Booth1999wt}, the Gibbs sampling of \cite{karim1992}, the quasi-likelihood (PQL) approach of \cite{breslow1993} and the standard Laplace approximation.}
  \label{tab:salamandercont}
  \end{table}
  
The approximate MLE obtained from the improved Laplace approximation with either exact or approximate conditional minima, the Monte Carlo Expectation-Maximisation (MC-EM) algorithm of \cite{Booth1999wt}, the quasi-likelihood approach of \cite{breslow1993}, the standard Laplace approximation and the posterior mean taken with the Gibbs sampling proposed by \cite{karim1992} are illustrated in Table~\ref{tab:salamandercont}. The standard Laplace approximation underestimates the variance parameters \citep[see also][]{shun1997another}. The estimate of $\beta$ and that of the variance parameters based on both versions of the improved Laplace approximation are quite similar to those of the MC-EM procedure of \cite{Booth1999wt} \citep[see also][]{sung2007}. However, compared to MC-EM, the proposed method does not require tuning from the practitioner. 

We notice that the approximate MLE based on the improved Laplace method with exact conditional minima is found within 9.5 minutes and after 14 iterations. Using approximate conditional minima, the approximate MLE is located within  3 minutes and after 15 iterations. 

We can use \eqref{eq:mlGLMM} also for conducting full likelihood-based inference. For instance, in the case of the Salamander data analysed jointly, let us consider profile likelihood-based confidence intervals for $\sigma_f^2$ and $\sigma_m^2$. Figure~\ref{fig:4} depicts the aforementioned relative profile likelihoods, obtained with the standard Laplace and with the improved Laplace with approximate conditional minima.
\begin{figure}
\centering
\includegraphics[scale = 0.63, angle=-90]{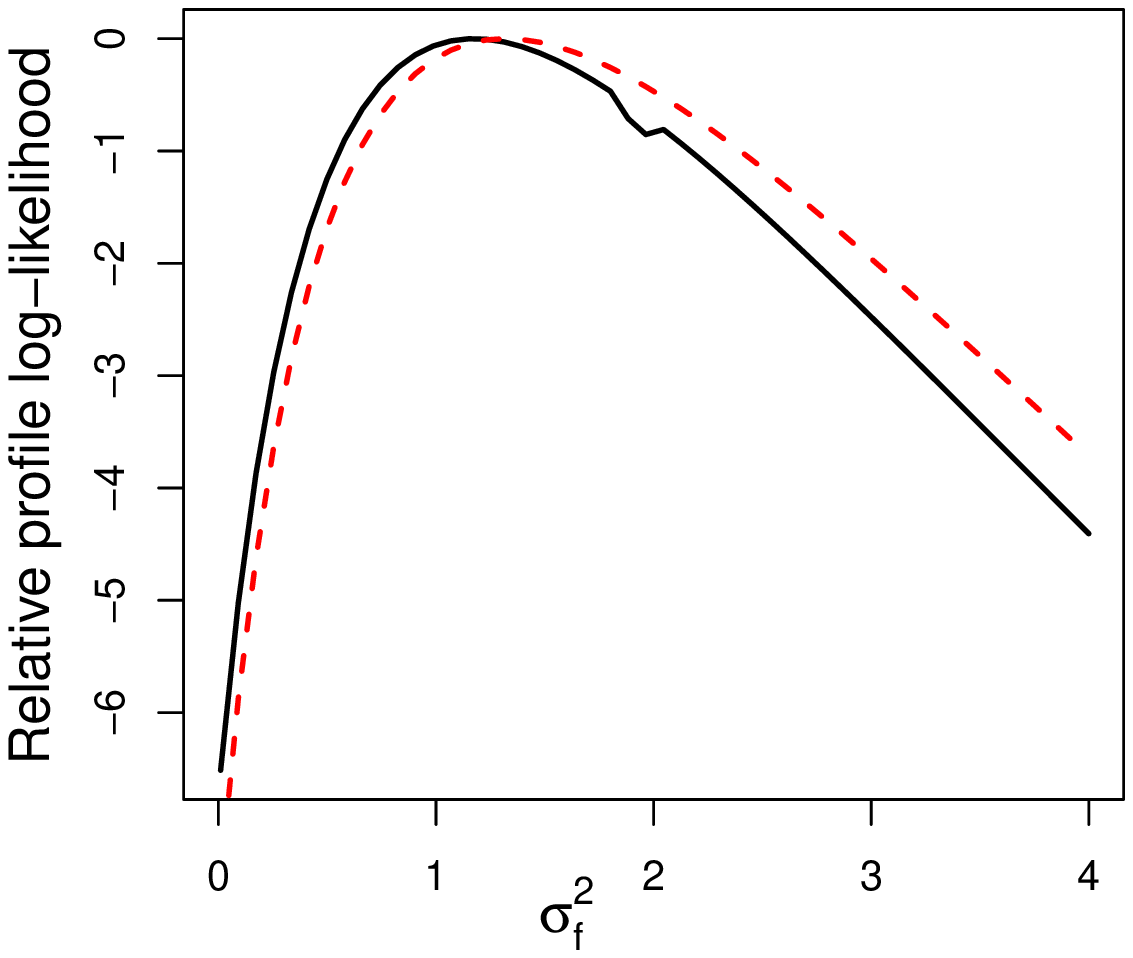}
\includegraphics[scale = 0.63, angle=-90]{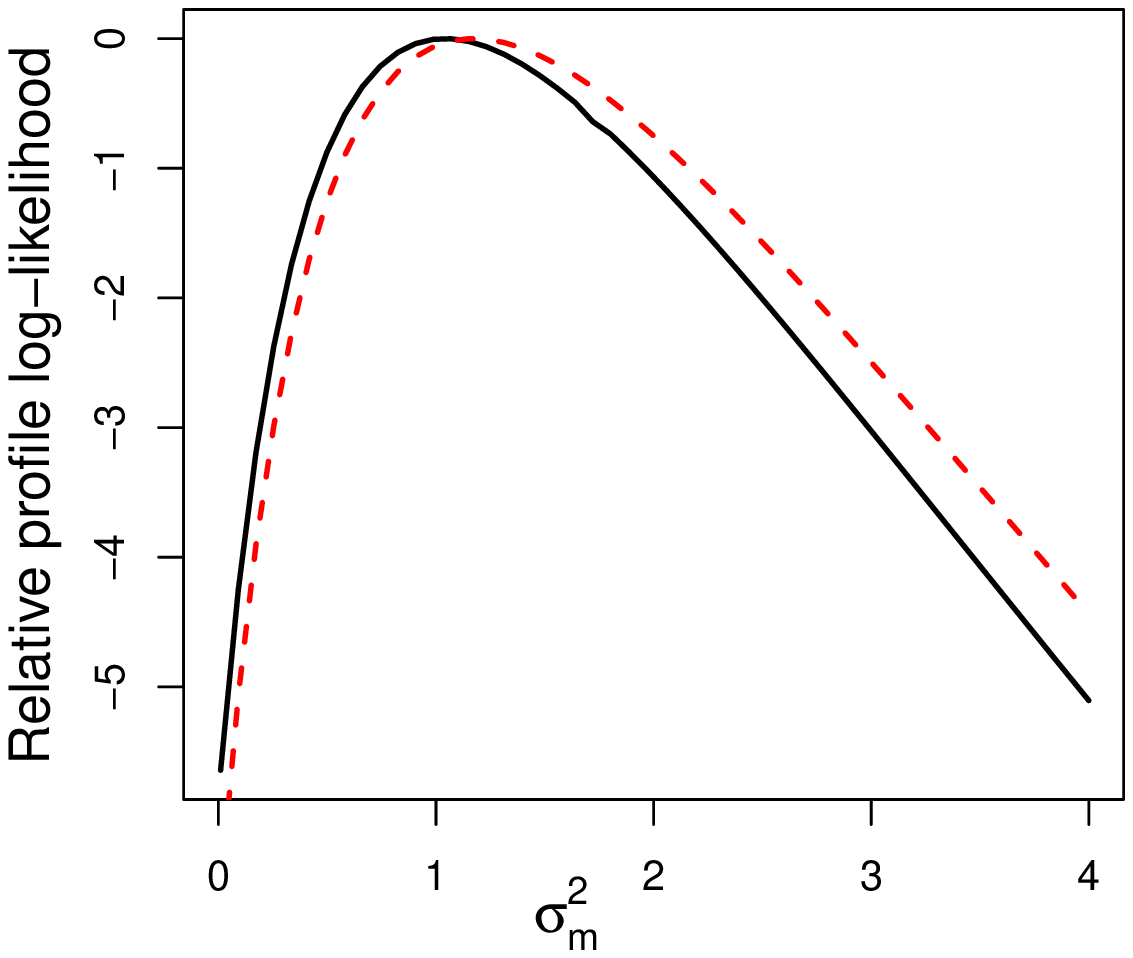}
\caption{Salamander mating data analysed jointly. Approximate relative profile log-likelihoods for the variance components obtained with the standard Laplace approximation (black continued) and with the improved Laplace approximation with approximate conditional minima (red dashed).}
\label{fig:4}
\end{figure}
From this plot we notice that standard Laplace-based profile likelihoods for the variance parameters are narrower than those based on the improved Laplace approximation. For instance, the 0.95 confidence interval found by inverting the profile likelihood of $\sigma_f^2$ and $\sigma_m^2$ with the standard Laplace approximation are (0.38, 2.7) and (0.31, 2.46) while those based on the improved Laplace approximation are (0.46, 2.98) and (0.38, 2.69).

\subsection{GLMM with spatial random effects}
\label{sec:rongelap}
We consider the application a Poisson geostatistical model to the Rongelap dataset \citep{diggle1998model}. This dataset reports counts $y_i$ on radionuclide concentration over the length of time $t_i$, at the spatial location $s_i$, for  $i=1,\ldots,157$ different locations in Rongelap Island. On the basis of the theory of radioactive emissions, the count $Y_i$ at the $n$ locations can be treated approximately as realisations of independent random variables with mean $\lambda(s_i) = t_i\eta(s_i)$, where $\eta(s_i)$ measures the radioactivity at location $s_i$, $i=1,\ldots,n$. See \cite{diggle1998model} and reference therein for further details.

For these data, \cite{diggle1998model} propose the following geostatistical model

\begin{eqnarray}
Y_i\,\vert\,\beta_0,\,\Sigma\,, u(s_i) &\sim\,&	\text{Poisson}(t_i\eta(s_i)),\label{eq:glmm_sp}
\end{eqnarray}

\[
\log\eta(s_i) = \beta_0 + u(s_i)
\]
where, $(U(s_1),\ldots,U(s_n)\, \sim\, N_n(0,\Sigma)$, are spatial random effects, which are marginally normally distributed with mean zero and covariance matrix $\Sigma$. Typically, estimation of a full $\Sigma$ is not possible, unless we place a proper prior on it, and some structure has to be imposed on it. Here we assume the exponential model, which implies that

\[
\Sigma_{ij} = \text{cov}(U(s_i),U(s_j) = \sigma^2\exp\{-||s_i-s_j||_2/\alpha\},
\]
where $\sigma^2$ is a variance parameter, $\alpha$ controls the correlation function. In our computations the distance matrix of the locations is divided by 100 in order to avoid numerical overflow problems in the computation of $\Sigma$.

The aim is to fit model \eqref{eq:glmm_sp} to the Rongelap data by maximum likelihood estimation, where the parameter of interest is $(\beta_0, \sigma^2, \alpha)$. The marginal likelihood of $(\beta_0, \sigma^2, \alpha)$ can be recast in the form of (2), with $f_u(\cdot;\theta_u)$ given by the multivariate normal distribution with mean zero and the covariance matrix controlled by $\theta_u=(\sigma^2,\alpha)$ and with $L(\theta;u,y)$ the likelihood given by the conditional distribution of $y$ given $u$ and  $\theta=\beta_0$.

To approximate the marginal likelihood we use the Laplace method and the improved Laplace approximation with approximate conditional minima. For comparison purposes, we approximate (2) also by importance sampling as proposed by \cite{sung2007}. Comparison with INLA in this case is not possible as the \texttt{R-INLA} package provides only numerical approximations to the marginal posterior distributions and not to the full marginal likelihood (2). We use the multivariate Student's $t$-distribution as importance density. The location of the importance density is fixed at the mode of the conditional density of $u$ given $y$ and $(\beta_0,\sigma^2,\alpha)$ and the scale matrix of the importance density is fixed at the Hessian matrix of the negative logarithm of the conditional density of $u$ given $y$ and $(\beta_0,\sigma^2,\alpha)$. The IS approximation of (2) for a fixed $(\beta_0,\sigma^2,\alpha)$ is 

\[
L_{IS}(\beta_0,\sigma^2,\alpha;y) = m^{-1}\sum_{j=1}^m 	\frac{L(\beta_0;u(s)_j,y)f_{u(s)}(u(s)_j; \Sigma)}{\tilde f(u(s)_j)}\,,
\]
where $u(s)_j$ is the $j$th random vector drawn from the importance distribution $\tilde f(\cdot)$, for $j=1,\ldots,m$ and $m$ is the overall number of random draws. To have an importance density with heavy tails we fix the degrees of freedom to 5. Furthermore, we fix the random seed in order to obtain a smooth approximation for the likelihood function. An issue with the Monte Carlo approximation is that the resulting estimate is subject to stochastic variability. To take this into account we consider $m=2\times10^4$ draws and compute the approximate MLE at 50 different seeds. Increasing $m$ would give more stable results but at the cost of higher computing time. The final estimates are obtained by averaging the 50 approximate MLEs and the Monte Carlo standard error is also computed from these replications.

 \begin{table}
 \centering
  \begin{tabular}{lccc}
  \hline
  \hline
  &\multicolumn{3}{c}{Methods}\\
  \cline{2-4}
 & Laplace & Improved  & Importance   \\
  Approximate MLE &  &  Laplace  & Sampling (3$\times$ $\text{SE}_{\text{MC}}$)  \\
 \hline
 $\beta_0$ & 1.83 & 1.83 & 1.83 (3.6$\times 10^{-6}$) \\
 $\sigma^2$    & 0.224  & 0.302 & 0.296 (2.2$\times 10^{-4}$)\\
  $\alpha$ &  0.081 & 0.104 & 0.103 (4.5$\times 10^{-5}$)\\
  \hline
  \hline
  \end{tabular}
  \caption{Rongelap data. Improved Laplace approximation (with approximate conditional minima) the standard Laplace approximation and importance sampling estimates. For the IS method, the point estimates are obtained by averaging over 50 replications obtained with 50 different random seeds; the quantity in parenthesis gives $3\times$SE$_{\text{MC}}$, where SE$_{\text{MC}}$ is the Monte Carlo standard error given by the standard deviation divided by $\sqrt{50}$.}
  \label{tab:rongelap}
  \end{table}

Results, shown in Table~\ref{tab:rongelap}, highlight that the standard Laplace approximation tend to underestimate both $\sigma^2$ and $\alpha$ as compared to the IS approximation, here treated as more trustworthy. On the other hand, the improved Laplace approximation gives similar results to the IS approximation. 
\begin{figure}
\centering
\includegraphics[width=8cm, height=8cm, angle=-90]{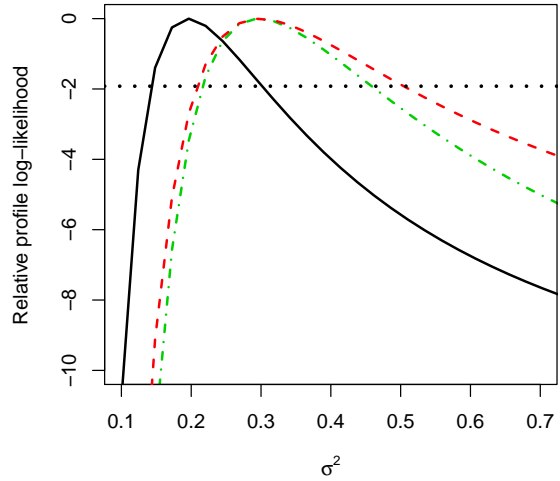}
\caption{Rongelap data. Approximate relative profile log-likelihoods for $\sigma^2$ obtained with the standard Laplace approximation (black continued), with the improved Laplace approximation with approximate conditional minima (red dashed) and with IS (green dot-dashed). The horizontal dotted line is the level which gives the 0.95 confidence interval for $\sigma^2$ based on the log-profile likelihood ratio statistic.}
\label{fig:5}
\end{figure}
Such a behaviour is perhaps more clear-cut if we look at the relative profile log-likelihood function of $\sigma^2$, depicted in Figure~\ref{fig:5} (the plots for $\alpha$ are similar and are omitted). This figure highlights that the standard Laplace approximation may produce misleading frequentist inference on the variance parameters, in terms of both, point estimation and profile likelihood-based interval estimation.

\section{Discussion} 
%Laplace's method for integrals is still widely used in both Bayesian and frequentist applications. However, when the sample size is small or when the dimension of the integral is high, the second-order Laplace approximation may deliver inaccurate results or it may even be asymptotically not valid as with the Salamander data. In this paper we proposed an improved Laplace approximation, and showed its superiority with respect to the standard Laplace method. Moreover, extensive numerical illustrations indicated that our method performs comparably well with other existing gold standard methods, which are computationally more demanding.

Although largely improving over the standard Laplace approximation, the proposed method is guaranteed to work only if the integrand is unimodal, with the mode being inside the domain of integration. This is because, if $h(\cdot)$ has either multiple minima or the minimum is not inside the domain of integration, then the determinants of blocks of its Hessian matrix may not be positive definite and the computation of \eqref{eq:iLAmarg} and \eqref{eq:condfree} may break down. A possible way to deal with multimodal integrands is through a mixture of Laplace approximations, e.g. one Laplace approximation for each of mode, provided they can all be found. However, these issues are open problems and are left for future work. 

A convenient feature of the proposed method is that the $d$ integrals can be easily computed in parallel. The numerical re-normalisations are an additional and difficult-to-quantify source of error. However, scalar numerical integration via carefully chosen adaptive quadratures is in general extremely accurate. 

The main computational burden of the method is due to conditional minimisations and Hessian determinants. Both can be greatly simplified by considering analytical first and second-order derivatives of $h(\cdot)$. This is the strategy adopted in the \texttt{iLaplace} package and throughout the examples. An alternative to analytical differentiation is the automatic differentiation, which provides on-line function differentiation during its evaluation \citep[see, e.g.,][]{david2012ad}. However, since automatic differentiation requires further programming efforts, we have not tried it in our package, though we plan to explore this possibility in future versions. 

The version of the method with approximate conditional minima showed good performance and significant savings in terms of computing time in the examples considered. Another alternative to this could be to use approximate conditional minima as starting points for the computation of the exact ones. Although this would speed-up the computation of conditional minima, the method might not be as fast as when using approximate conditional minima in place of the actual one. %We remark, however, that the use of approximate minima such as \eqref{eq:appmode}, in the Laplace approximation has been studied by \cite{miyata2004fully}, who also derived formulas for approximating posterior normalising constants and posterior moments. Formulas for marginal posterior densities with approximate minima can also be derived from results of \cite{miyata2004fully}. However, this is outside the scope of this work. 

From a practical perspective, the improved Laplace approximation requires the integrand to be concave and unimodal but not necessarily symmetric or with Gaussian tails, though further assumptions are required to guarantee its asymptotic properties. In our experience, the standard Laplace approximation tends to work poorly when many variables of the integrand lay on the positive subset of the real numbers or when the dimensionality of the integrand increases with the sample size. Indeed, despite applying logarithmic transformations, such variables may still lead to asymmetric or heavy-tailed integrands. While in Bayesian applications $H(\cdot)$ may not always be unimodal, in GLMM it is often unimodal. In many instances, with independent random effects, the standard Laplace approximation or numerical integration with a few quadrature points are accurate enough for practical purposes. Indeed standard GLMM can now be fitted quite accurately by available \texttt{R} packages such as \texttt{lme4}. However, in models with complicated, dependent and/or crossed random effects, Laplace's method may perform poorly, and numerical integration may require a large number of quadrature points, hence leading to a higher computational overhead. Our method seems particularly suited for these contexts, as was also demonstrated by the examples of {Sections~\ref{sec:salam} and \ref{sec:rongelap}}.

Finally, the improved Laplace approximation can be used to compute slightly modified KL divergences that arise in the variational approximation framework \citep[see, e.g.,][]{ormerod2010explaining}. In this context, the method can be useful for extending the usual Gaussian variational approach to the use of more flexible and non-conjugate densities, such as the skew-$t$ \citep{azzalini2003distributions}. This and the extension of the method to cases in which the mode lies outside the integration region are under investigation.

\bibliographystyle{biometrika}
\bibliography{biblio}
%
%%%%%%%%%%%%%%%%%%%%%%%%%%%%%%%%%%%%%%%%%%%%%%%%%%%%%%%%%%%%%%%%%
\end{document}